\documentclass[11pt]{article}
\usepackage{titletoc}

\usepackage[utf8]{inputenc}
    \usepackage[dvipsnames]{xcolor}
    \usepackage[T1]{fontenc}
	\usepackage[utf8]{inputenc}
	\usepackage[english]{babel}
	\usepackage{amsfonts}
	\usepackage{amssymb}
	    \usepackage{amsmath}
	
	\usepackage{enumerate}
	\usepackage{bbm}
	\usepackage{mathtools}
	\usepackage{subfig}
	\usepackage{booktabs}	
	\usepackage{setspace}
	\usepackage{graphicx}
	\usepackage{fullpage}
	
	\usepackage{soul}
	\usepackage{xspace}
	\usepackage{csquotes}
	\usepackage[margin = 1in]{geometry}
	\usepackage{placeins}
	\usepackage[footfont=normalsize,font=normalsize]{floatrow}
	\usepackage{lscape}
	\usepackage{longtable}
	\usepackage{mathabx}
    \usepackage{accents}
    \usepackage{float}

    \usepackage{tikz}
    \usetikzlibrary{calc,intersections}

	\usepackage{natbib}

\definecolor{cmt}{HTML}{408080}
	\usepackage{hyperref}
	\hypersetup{
    colorlinks=true,
    linkcolor=Blue,
    filecolor=magenta,      
    urlcolor=blue,
    citecolor = Blue
}

    \usepackage{rotating}
	\usepackage[framemethod=tikz]{mdframed}
	\usepackage{todonotes}
	\presetkeys{todonotes}{color=black!5,inline}{} %

    \usepackage{tcolorbox}
    \newtcbox{\feedback}{nobeforeafter,colframe=black,colback=white,boxrule=0.5pt,arc=2pt,
      boxsep=0pt,left=2pt,right=2pt,top=2pt,bottom=2pt,tcbox raise base}

\usepackage{array}
\newcolumntype{L}[1]{>{\raggedright\let\newline\\\arraybackslash}m{#1}}
\newcolumntype{C}[1]{>{\centering\let\newline\\\arraybackslash\hspace{0pt}}m{#1}}
\newcolumntype{R}[1]{>{\raggedleft\let\newline\\\arraybackslash\hspace{0pt}}m{#1}}

\usepackage{ragged2e}
\newlength\ubwidth

\newcommand\numberthis{\addtocounter{equation}{1}\tag{\theequation}}

	\renewcommand{\P}{\mathop{}\!\textnormal{P}}
	\newcommand{\E}{\mathop{}\!\textnormal{E}}

	\newcommand{\parens}[1]{\left(#1\right)}

	\newcommand{\R}{\mathbb{R}}

	\newcommand{\reals}{\mathbb{R}}

    \newcommand{\indep}{\perp \!\!\! \perp}

    \renewcommand{\th}{\hat{\theta}}

\newcommand{\arcsinh}{\operatorname{arcsinh}}
\usepackage{xr-hyper}

\usepackage[thm]{Preambles/macros}
\usepackage{thm-restate}
\usepackage{pdflscape}
\usepackage{afterpage}
\usepackage{tabularx}

\crefname{appendixfigure}{Appendix Figure}{Appendix Figures}
\crefname{appendixtable}{Appendix Table}{Appendix Tables}
\newcommand\DoToC{%
  \startcontents
  \printcontents{}{0}{\textbf{Contents}\vskip1em\hrule\vskip1em}
  \vskip1em\hrule\vskip5pt
}
\usepackage{makecell}

\newgeometry{margin=1in}
\renewcommand{\E}{E}
\renewcommand{\P}{P}
\makeatletter
\AtBeginDocument{%
  \def\thmhead#1#2#3{%
    \thmname{#1}\thmnumber{\@ifnotempty{#1}{ }\@upn{#2}}%
    \thmnote{ {\normalfont (#3)}}}%
}
\makeatother

\let\oldFootnote\footnote
\newcommand\nextToken\relax

\renewcommand\footnote[1]{%
    \oldFootnote{#1}\futurelet\nextToken\isFootnote}

\newcommand\isFootnote{%
    \ifx\footnote\nextToken\textsuperscript{,}\fi}

\usepackage{clipboard}
\newclipboard{output-draft}

\title{
Logs with zeros? Some problems and solutions%
\thanks{An earlier draft of this paper was titled ``Log-like?
Identified ATEs defined with zero-valued outcomes are (arbitrarily) scale-dependent.'' We
thank Isaiah Andrews, Kirill Borusyak, Jonathan Cohn, Amy Finkelstein, Edward Glaeser,
Nick Hagerty, Peter Hull, Jetson Leder-Luis, Erzo Luttmer, Giovanni Mellace, John Mullahy, Edward Norton,
David Ritzwoller, Brad Ross, Pedro Sant'Anna, Jesse Shapiro, Neil Thakral, Casey Wichman, and seminar
participants at BU, Georgetown, Harvard/MIT, Southern Denmark University, Vanderbilt,
Stanford, UC-Irvine, UCLA, UCSD, and the SEA annual conference for helpful comments and
suggestions. Bruno Lagomarsino provided superb research assistance.}}

\author{Jiafeng Chen \\ Harvard Business School \\ Department of Economics, Harvard
University \and Jonathan Roth \\ Department of Economics, Brown University}

\date{November 15, 2023}

\begin{document}

\maketitle

\begin{abstract} 
\Copy{abstract}{When studying an outcome $Y$ that is weakly-positive but can equal zero (e.g. earnings), researchers frequently estimate an average treatment effect (ATE) for a ``log-like'' transformation that behaves
like $\log(Y)$ for large $Y$ but is defined at zero (e.g. $\log(1+Y)$,  $\arcsinh(Y)$). We argue that ATEs for log-like transformations should not be interpreted as
approximating percentage effects, since unlike a percentage, they depend on the units of the outcome. In fact, we show that if the treatment affects the extensive margin, one can obtain a treatment effect of any magnitude simply by re-scaling the units of $Y$ before taking the log-like transformation. This arbitrary unit-dependence
arises because an individual-level percentage effect is not well-defined for individuals
whose outcome changes from zero to non-zero when receiving treatment, and the units of the
outcome implicitly determine how much weight the ATE for a log-like transformation places on the extensive margin. We further establish a trilemma: when the outcome can equal zero, there is no treatment
effect parameter that is an average of individual-level treatment effects,
unit-invariant, and point-identified. We discuss several alternative approaches that
may be sensible in settings with an intensive and extensive margin, including (i)
expressing the ATE in levels as a percentage (e.g. using Poisson regression), (ii)
explicitly calibrating the value placed on the intensive and extensive margins, and (iii)
estimating separate effects for the two margins (e.g. using Lee bounds). We illustrate
these approaches in three empirical applications.}

\end{abstract}

\clearpage
\section{Introduction}
\Copy{firstp}{
When the outcome of interest $Y$ is strictly positive, researchers often estimate an
average treatment effect (ATE) in logs of the form $E_P[ \log(Y(1)) - \log(Y(0)) ]$, which
has the appealing feature that its units approximate percentage changes in the
outcome.\footnote{That is, $\log(Y(1)/Y(0)) \approx \frac{Y(1) - Y(0) }{Y(0)}$ when
$Y(1)/Y(0) \approx 1$.} A practical challenge in many economic settings, however, is that
the outcome may sometimes equal zero, and thus the ATE in logs is not well-defined. When
this is the case, it is common for researchers to estimate ATEs for alternative
transformations of the outcome such as $\log(1+Y)$ or $\arcsinh(Y) =
\log\pr{\sqrt{1+Y^2} + Y}$, which behave similarly to $\log(Y)$ for large values of $Y$
but are well-defined at zero. The treatment effects for these alternative transformations
are typically interpreted like the ATE in logs, i.e. as (approximate) average percentage
effects. For example, among the 11 papers published in the \textit{American Economic
Review} since 2018 that interpret a treatment effect for $\arcsinh(Y)$, all but one
interpret the result as a percentage effect or elasticity.\footnote{We found 17 papers
overall using $\arcsinh(Y)$ as an outcome variable, of which 11 interpret the units; see
\cref{tbl:selected-quotes-aer}.} }

The main point of this paper is that identified ATEs that are well-defined with
zero-valued outcomes should not be interpreted as percentage effects, at least if one imposes the
logical requirement that a percentage effect does not depend on the baseline units in
which the outcome is measured (e.g. dollars, cents, or yuan).

Our first main result shows that if $m(y)$ is a function that behaves like $\log(y)$ for
large values of $y$ but is defined at zero, then the ATE for $m(Y)$ will be
\textit{arbitrarily sensitive} to the units of $Y$. Specifically, we consider continuous,
increasing functions $m(\cdot)$ that approximate $\log(y)$ for large values of
$y$ in the sense that $m(y)/\log(y) \to 1$ as $y \to \infty$. The
common $\log(1+y)$ and $\arcsinh(y)$ transformations satisfy this property. We show
that if the treatment affects the extensive margin (i.e. $\P(Y(1) = 0) \neq
\P(Y(0) =0)$), then one can obtain any magnitude for the ATE for $m(Y)$ by rescaling the
outcome by some positive factor $a$. It is therefore inappropriate to interpret the ATE
for $m(Y)$ as a percentage effect, since a percentage is inherently a unit-invariant
quantity, while the ATE for $m(Y)$ depends arbitrarily on the units of $Y$.

The intuition for this result is that a ``percentage'' treatment effect is not
well-defined for an individual for whom treatment increases their outcome from zero to a
positive value. For example, in our application to \citet{carranza2022job} in \cref{sec:
empirical}, the treatment induces more people to have positive hours worked. The
percentage change in hours is then not well-defined for individuals who would work
positive hours under the treatment condition but zero hours under the control condition.
Any average treatment effect that is well-defined with zero-valued outcomes must therefore
implicitly assign a value for a change along the extensive margin. For logarithm-like
transformations $m (\cdot)$, the importance of the extensive margin is determined
implicitly by the units of $Y$. To see why this is the case, consider an individual who
works positive hours only if they are treated, so that $Y(1)>0$ and $Y(0)=0$. Their
treatment effect for the transformed outcome $m(Y)$ is $m(Y(1)) - m(0)$, which becomes
larger if the units of $Y$ are re-scaled by some $a > 1$, e.g. if we convert from weekly hours worked to
yearly hours worked. When the treatment has an extensive margin effect, the ATE for $m(Y)$
can thus be made large in magnitude by re-scaling $Y$ by a large factor $a$. By contrast, if we
re-scale $Y$ by a small factor $a \approx 0$, such that the resulting outcomes are close to zero, then $m(Y)
\approx m(0)$, and so the ATE for $m(Y)$ will be small. By varying the units of the
outcome, we can thus obtain any magnitude for the ATE for $m(Y)$.

Our theoretical results also imply that if we re-scale the units of the outcome by a
finite factor $a>0$, the ATE for a log-like transformation $m(Y)$ will change by
approximately $\log(a)$ times the effect of the treatment on the extensive margin. This
result implies that sensitivity analyses that explore how the estimated ATE for $m(Y)$
changes with finite changes in the units of $Y$---or equivalently,
how the ATE for $\log(c+Y)$ changes with the constant $c$---are essentially indirectly
measuring the size of the treatment effect on the extensive margin. %

We illustrate the practical importance of these results by systematically replicating
recent papers published in the \emph{American Economic Review} that estimate treatment
effects for $\arcsinh$-transformed outcomes. In line with our theoretical
results, we find that treatment effect estimates using $\arcsinh(Y)$ are sensitive to changes in the
units of the outcome, particularly when the extensive margin effect is large. In half of
the papers that we replicated, multiplying the original outcome by a factor of 100 (e.g.
converting from dollars to cents) changes the estimated treatment effect by more than
100\% of the original estimate. We obtain similar results using $\log(1+Y)$ instead of $\arcsinh(Y)$.

What, then, are alternative options in settings with zero-valued outcomes? Our second main
result delineates the possibilities. We show that when there are zero-valued outcomes,
there is no treatment effect parameter that satisfies all three of the following
properties:
\begin{enumerate}[label=(\alph*)]
    \item 
    The parameter is an average of individual-level treatment effects, i.e. takes the form $\theta_g = E_P[
    g(Y(1),Y(0))]$, where $g$ is increasing in $Y(1)$.
    
    \item 
    The parameter is invariant to re-scaling of the units of the outcome (i.e. $g(y_1,y_0) = g(a y_1, a y_0)$). 
    
    \item 
    The parameter is point-identified from the marginal distributions of the potential outcomes. 
\end{enumerate}
\noindent This ``trilemma'' implies that any target parameter that is well-defined with zero-valued
outcomes must necessarily jettison at least one of the three properties above.
\Copy{targetparamintro}{Of course, the choice of target parameter should depend on the
economic question of interest. Which of the three properties the researcher prefers to
forgo will thus generally depend on their context-specific motivation for using a log-like
transformation in the first place.}

To that end, \cref{sec: recommendations} highlights a menu of parameters that may be attractive
depending on the researcher's core motivation. We first consider the case where the researcher is interested
in obtaining a causal parameter with an intuitive ``percentage'' interpretation. In this case, it may be natural to consider a parameter outside of the class of individual-level averages of the form $E_P[g(Y(1),Y(0))]$. One prominent option is $\theta_{ \text{ATE}\%} =
\frac{E[Y(1)-Y(0)]}{E[Y(0)]},$ the ATE in levels as a percentage of the baseline mean, which in many cases can be estimated via Poisson regression \citep{silva_log_2006,wooldridge2010econometric}. The researcher might also consider alternative normalizations of the outcome that lead to intuitive units, e.g. expressing the outcome in per-capita units or converting it to a rank with respect to some reference distribution. Next, we suppose
the researcher would like to capture concave preferences over the outcome; for example,
the researcher might consider income gains to be more meaningful for individuals who are
initially poor. In this case, it is natural to directly specify how much the researcher
values a change along the extensive margin relative to the intensive margin---e.g., that a
change from 0 to 1 is worth an $x$ percent change along the intensive margin. Finally,
suppose the researcher is interested in separately understanding the effects of the
treatment along both the intensive and extensive margins. In this case, the researcher may
target separate parameters for the two margins---e.g., $E[ \log(Y(1)) -
\log(Y(0)) \mid Y(1) > 0,Y(0) > 0]$, the average effect in logs for individuals with
positive outcomes under both treatments, captures the intensive margin. Separate effects
for the two margins are not generally point-identified, but can be can be bounded using
the method in \citet{lee_training_2009} or point-identified with additional assumptions
\citep{zhang_evaluating_2008, zhang_likelihood-based_2009}. 

\Cref{sec: empirical} provides a blueprint for estimating these alternative parameters in
practice by applying our recommended approaches to three recent empirical applications,
including a randomized controlled trial (RCT) \citep{carranza2022job}, a
difference-in-differences (DiD) setting \citep{sequeira_corruption_2016}, and an instrumental
variables (IV) setting \citep{berkouwer2022credit}.
 
\paragraph{Related work.} The use of log-like transformations for dealing with zero-valued
outcomes has a long history. The use of the $\log(1+Y)$ transformation dates to at least
\citet{williams_use_1937}, while \citet{bartlett_use_1947} considers both the $\log(1+Y)$
and inverse hyperbolic sine transformations.\footnote{\citet{bartlett_use_1947} proposes
using $\arcsinh(\sqrt{Y})$.} \Copy{bw}{More recent papers by
\citet{burbidge_alternative_1988} and \citet{bellemare_elasticities_2020}, among others,
provide results for $\arcsinh(Y)$ that are frequently cited in economics
papers using this transformation.}\footnote{\citet{mackinnon_transforming_1990} propose
transformations of the form $\arcsinh(y \zeta)/\zeta$, where $\zeta$ is estimated by
assuming $\arcsinh (y\zeta)/\zeta$ is normally distributed conditional on covariates.}

Previous work has illustrated in simulations or selected empirical applications that
results for particular transformations such as $\log(1+Y)$ or $\arcsinh(Y)$ may be
sensitive to the units of the outcome \citep{aihounton_units_2021, de_brauw_income_2021}.
In concurrent work, \citet{mullahy_why_2023} show theoretically that the marginal effects
from linear regressions using $\log(1+Y)$ or $\arcsinh(Y)$ are sensitive to the scaling of
the outcome, with the the limits of the marginal effects approaching those of either a
levels regression or a (normalized) linear probability model, depending on whether the
units are made small or large. We complement this work by proving that scale-dependence is
a necessary feature of \textit{any} identified ATE that is well-defined with zero-valued
outcomes, and that the dependence on units is arbitrarily bad for transformations that
approximate $\log(Y)$ for large values of $Y$. Thus, it is not possible to fix the issues
with $\log(1+Y)$ or $\arcsinh(Y)$ by choosing a ``better'' transformation or using a
different estimator. We also complement previous empirical examples by providing a
systematic analysis of the sensitivity to scaling for papers in the $\emph{American
Economic Review}$ using $\arcsinh(Y)$.

Other work has considered the interpretation of regressions using $\arcsinh(Y)$
or $\log(1+Y)$ from the perspective of structural equations models, as opposed to the
potential outcomes model considered here. This literature has reached diverging
conclusions: For example, \citet{bellemare_elasticities_2020} conclude that coefficients
from $\arcsinh(Y)$ regressions have an interpretation as a semi-elasticity, while
\citet{cohn_count_2022} conclude that these estimators are inconsistent and advocate for
Poisson regression instead. \citet{thakral2023estimates} show that the semi-elasticities implied by OLS regressions using $\arcsinh(Y)$
or $\log(1+Y)$ are sensitive to scale; they recommend instead the use of power functions $Y^k$, which they show are the only transformations (besides $\log$) for which the implied semi-elasticities for OLS regressions are scale-invariant. In
\Cref{sec: structural equations}, we show that these diverging conclusions stem from the
fact that the structural equations considered in these papers implicitly impose different
restrictions on the potential outcomes---some of which are incompatible with zero-valued
outcomes---and consider different target causal parameters. This highlights the value of a
potential outcomes framework such as ours, which makes transparent what causal parameters
are identifiable and what properties they can have.

Finally, there is a long history in econometrics of explicitly modeling the intensive and
extensive margins in settings with zero-valued outcomes, such as
\citet{tobin_estimation_1958} and \citet{heckman_sample_1979}. Broadly
speaking, these methods impose parametric structure on the joint distribution of the
potential outcomes, which allows one to separate out the intensive and extensive margin
effects of a treatment (see \Cref{sec: structural equations} for technical details). Of
course, the parametric restrictions underlying these approaches may often be difficult to
justify in practice, which perhaps has contributed to the growth in the use of log-like
transformations in place of approaches that explicitly model the extensive margin. Our
paper shows that the presence of an extensive margin should not simply be ignored by taking a
log-like transformation. It also clarifies what parameters can be learned in such cases without imposing
restrictions on the joint distribution of the potential outcomes.

\subsection{Setup and notation}

Let $D \in \{0,1\}$ be a binary treatment and let $Y \in [0,\infty)$ be a weakly
positively-valued outcome.\footnote{The $\arcsinh$ transformation is sometimes used in settings where $Y$ can be negative. We impose that $Y \in [0,\infty)$, and thus do not consider this case. See \Cref{subsec: continuous treatment} for extensions
of our results to settings with continuous treatments.\label{foot:neg}} We assume that $Y = D Y(1) +
(1-D)Y(0)$, where $Y(1)$ and $Y(0)$ are respectively the potential outcomes under
treatment and control. We suppose that in some (sub-)population of interest, $(Y(1),Y(0)) \sim
P$ for some (unknown) joint distribution $P$. We denote the marginal distribution of
$Y(d)$ under $P$ by $P_{Y(d)}$ for $d=0,1$. We assume that neither $P_{Y(0)}$ nor
$P_{Y(1)}$ is a degenerate distribution at zero.

\section{Sensitivity to scaling for transformations that behave like \texorpdfstring{$\log(Y)$}{log(Y)}}
\label{sec: loglike results}

We first consider average treatment effects of the form $\theta = \E_P[ m(Y(1)) - m(Y(0))
]$ for an increasing function $m$. We note that $\theta$ corresponds to the ATE among the (sub-)population indexed by $P$; if $P$ refers to the sub-population of compliers for an instrument, for instance, then $\theta$ is the local average treatment effect (LATE), rather than the ATE in the full population. We are interested in how $\theta$ changes if we change
the units of $Y$ by a factor of $a$. That is, how does
\[\theta(a) = \E_P[ m(aY(1)) - m(aY(0)) ]\] depend on $a$? Setting $a=100$, for example,
might correspond with a change in units between dollars and cents. Of course, if $Y$ is strictly positive and $m(y) = \log(y)$, then $\theta(a)$ is the ATE in logs and does not depend on the value of $a$.

We consider ``log-like'' functions $m(y)$ that are well-defined at zero but behave like $\log(y)$ for large values of $y$, in the
sense that $m(y)/ \log(y) \to 1$ as $y \to \infty$. This property is satisfied by $\log
(1+y)$ and $\arcsinh(y)$, for example. Our first main result shows that if the treatment
affects the extensive margin, then $|\theta(a)|$ can be made to take any desired value
through the appropriate choice of $a$.

\begin{restatable}{prop}{anyvalue}
\label{prop: can get any value for ATE}
Suppose that:
\begin{enumerate}
    \item
    (The function $m$ is continuous and increasing) $m:[0,\infty) \to \reals$ is a
    continuous, weakly increasing function.
    
    \item 
    (The function $m$ behaves like $\log$ for large values) $m(y)/\log(y) \to 1$ as $y \to
    \infty$.

    \item
    (Treatment affects the extensive margin) $P(Y(1)=0) \neq P(Y(0) = 0)$.

    \item
    (Finite expectations) $E_{P_{Y(d)}}[ |\log(Y(d))| \mid Y(d) > 0 ] < \infty$ for
    $d=0,1$.\footnote{This assumption simply ensures that $\E_{P_{Y(d)}}[|m(aY(d))|
    \mid Y > 0]$ exists for all values of $a>0$.}
\end{enumerate}
Then, for every $\theta^* \in (0,\infty)$, there exists an $a>0$ such that $|\theta(a)| =
\theta^*$. In particular, $\theta(a)$ is continuous with $\theta(a) \to 0$ as $a \to 0$
and $|\theta(a)| \to \infty$ as $a \to \infty$.
\end{restatable}

\Cref{prop: can get any value for ATE} casts serious doubt on the interpretation of ATEs
for functions like $\log(1+Y)$ or $\arcsinh(Y)$ as (approximate) average percentage effects. While
a percent (or log point) is entirely invariant to the units of the outcome, \cref{prop: can get any value
for ATE} shows that, in sharp contrast, the ATEs for these transformations are \emph{arbitrarily}
dependent on units.

\subsection{Intuition for \cref{prop: can get any value for ATE}}
\label{sub:any_value_ate}

Loosely speaking, the result in \cref{prop: can get any value for ATE} follows from the
fact
that a ``percentage'' treatment effect is not well-defined for individuals who have $Y (0)
= 0$ but $Y(1) > 0$.\footnote{See \citet{delius_cash_2020} for an intuitive discussion of this difficulty in the context of the $\arcsinh(\cdot)$ transformation. They write, ``the concept of elasticity itself does not make sense with zeros'' (p. 21).} Any ATE that is well-defined with zero-valued outcomes
must implicitly determine how much weight to place on changes along the extensive margin
relative to proportional changes along the intensive margin.

When $m(Y)$ behaves like $\log(Y)$ for large values of $Y$, the importance of the
extensive margin is implicitly determined by the units of $Y$. For intuition, suppose that
we re-scale the outcomes so that the non-zero values of $Y$ are very large. Then for an
individual for whom treatment changes the outcome from zero to non-zero, the treatment
effect will be very large, since $m(Y(1)) \gg m(Y(0)) = m(0)$. Extensive margin treatment
effects thus have a large impact on the ATE when the values of $Y$ are made large. By
contrast, changing the units of $Y$ does not change the importance of treatment effects
along the intensive margin by much, since for $Y(1)> 0$ and $Y(0) >0$, we have that $m(Y
(1)) -
m(Y(0)) \approx \log(Y(1)/Y(0))$, which does not depend on the units of the outcome.

To see the roles of the extensive and intensive margins more formally, for simplicity
consider the case where $P(Y(1) = 0 , Y(0) > 0 ) = 0$, so that, for example, everyone who has
positive income without receiving a training also has positive income when receiving the
training.\footnote{A similar argument goes through without this restriction, but then there
are two extensive margins, one for individuals with $Y(1)>0=Y(0)$, and the other for those
with $Y(0)>Y(1)=0$.} Then, by the law of iterated expectations, we can write
\begin{align*}
\E[ m(aY(1)) - m(aY(0)) ] &= \P(Y(1)>0,Y(0)>0) \underbrace{\E_P[ m(aY(1)) - m(aY(0)) \mid
Y(1)>0,Y(0)>0  ]}_{\text{Intensive margin}} \\ &+ \P(Y(1)>0,Y(0) =0) \underbrace{\E_P[
m(aY(1)) - m(0) \mid Y(1)>0,Y(0)=0 ]}_{\text{Extensive margin}}.
\end{align*}
\noindent When $a$ is large, $m(ay) \approx \log(ay)$ for non-zero values of $y$, and thus
the intensive margin effect in the previous display is approximately equal to $E_P[
\log(Y(1)) -
\log(Y(0)) \mid Y(1)>0,Y(0)>0]$, the treatment effect in logs for individuals with
positive outcomes under both treatment and control. This, of course, does not depend on
the scaling of the outcome. However, the extensive margin effect grows with $a$, since
$m(aY(1)) \approx \log(a) + \log(Y(1))$ is increasing in $a$ while $m(0)$ does not depend on $a$.
Thus, as $a$ grows large, the ATE for $m(aY)$ places more and more weight on the extensive
margin effect of the treatment relative to the intensive margin. We can therefore make
$|\theta(a)|$ arbitrarily large by sending $a \to \infty$. By contrast, if $a \approx 0$,
then $m(aY(d)) \approx 0$ with very high probability, and thus the ATE for $m(aY)$ is
approximately equal to 0.

\Copy{nuisancezeros1}{It is worth emphasizing that the arbitrary scale-dependence
described in \cref{prop: can get any value for ATE} exists whenever the treatment affects
the probability that the outcome is zero, regardless of whether the extensive margin is of direct economic interest
or not.\footnote{Without an extensive margin, ATEs for
transformations $m(\cdot)$ defined at zero still exhibit scale-dependence, though perhaps
not arbitrarily so. See \cref{subsec: no zeros discussion} below for further discussion.} In some settings, the presence of zeros may correspond to a discrete economic
choice (e.g. not participating in the labor market), and thus may be of direct interest.
In other settings---for example, if the outcome is a yearly count of publications which is
sometimes zero for idiosyncratic reasons---the extensive margin may be a ``nuisance''
rather than a direct economic object of interest.\footnote{One setting where nuisance
zeros may arise is when the observed outcome $Y$ is actually a mis-measured version of the
true economic object of interest. For example, publications $Y$ may be a noisy measure of
true researcher productivity $Y^* > 0$. One possible remedy in this setting is to model
the measurement error to recover the treatment effect on $Y^*$ rather than on $Y$. In a
similar vein, \citet{gandhi_estimating_2023} models the measurement error in product
shares in demand estimation, which are sometimes zero in finite samples.} The result in
\cref{prop: can get any value for ATE} highlights that regardless of the source of the
zeros, an ATE for a log-like transformation is not interpretable as a percentage, since
the presence of the extensive margin effect makes it arbitrarily dependent on the units.
Indeed, a percentage effect is not a well-defined for individuals moving from zero to
non-zero outcomes. Whether the zeros correspond to a discrete economic choice or not will
be relevant, however, when considering the choice of alternative target parameter, a topic
we return to in \cref{sec: recommendations} below.}

\subsubsection{Intuition for the special case of $\log(1+Y)$\label{subsubsec: log1plus intuition}}

We can also develop some intuition for \Cref{prop: can get any value for ATE} by
considering the special case where $m(y) = \log(1+y)$. In that case, we have that
\begin{equation}
\theta(a) = E[\log(1 +aY(1)) - \log(1 + aY(0))] = E\bk{\log\left( \dfrac{1 + aY(1)}{1 + aY(0)} \right) }. \label{eqn: theta for log1plus}
\end{equation}
\noindent Note that 
$$\lim_{a \to \infty} \log\left( \dfrac{1 + aY(1)}{1 + aY(0)} \right) =   \begin{cases}
    \log\left(\frac{Y(1)}{Y(0)}\right) & \text{ if } Y(1)>0,Y(0)>0 \\
    0 & \text{ if } Y(1)=0,Y(0) =0 \\
    \infty &\text{ if } Y(1)>0,Y(0)=0 \\
    -\infty & \text{ if } Y(1)=0,Y(0)>0.
    \end{cases}$$
\noindent We thus see that the term inside the expectation in \eqref{eqn: theta for
log1plus} diverges to $\infty$ for individuals with $Y(1)>0,Y(0)=0$, and likewise diverges
to $-\infty$ when $Y(1)=0,Y(0)>0$. If on average the extensive margin effect is positive,
then there are more individuals for whom the limit is $+\infty$ rather than $-\infty$, and
thus (under appropriate regularity conditions) the ATE diverges to $\infty$.
Analogously, if the extensive margin effect is negative, then the ATE diverges to $-\infty$. Hence, we see that the magnitude of the ATE for $\log(1+aY)$
diverges as $a \to \infty$ when the average effect on the extensive margin is non-zero. By
contrast, as $a \to 0$, $\log(1 + aY(d)) \to \log(1) = 0$ for both $d=0$ and $d=1$, and
thus the treatment effect converges to 0. \Cref{prop: can get any value for ATE} shows
that this dependence on units occurs for \emph{any} log-like transformation, not just
$\log(1+Y)$, and thus this issue cannot be fixed by choosing a different log-like
transformation ($\log(c+Y)$, $\arcsinh(Y)$,  $\arcsinh(\sqrt{Y})$, etc.)

\subsection{Additional remarks and extensions}

\begin{rmk}[ATEs for $\log(c+Y)$] \label{rmk: log c plus y}
In some settings, researchers consider the ATE for $\log(c+Y)$ and investigate
sensitivity to the parameter $c$. Observe that $\log(1 + aY) = \log(a (1/a +Y)) = \log(a)
+ \log( 1/a + Y)$, and thus the ATE for $\log(1 + aY)$ is equal to the ATE for $\log(1/a +
  Y)$. Hence, varying the constant term for $\log(c+Y)$ is equivalent to varying the
  scaling of the outcome when using $m(y)= \log(1+y)$. \Cref{prop: can get any value for ATE}
  thus implies that if treatment affects the extensive margin, one can obtain any
  desired magnitude for the ATE for $\log(c+Y)$ via the choice of $c$. In particular, the ATE for $\log(c+Y)$ grows large in magnitude as $c \to 0$, and small as $c \to \infty$. 
\end{rmk}

\begin{rmk}[Finite changes in scaling] \label{rmk: finite changes in scale}
\Cref{prop: can get any value for ATE} shows that any magnitude of $|\theta(a)|$ can be
achieved via the appropriate choice of $a$. How much does $\theta(a)$ change for finite
changes in the scaling $a$? \cref{prop: log a rate} in the appendix shows that the change
in the ATE from multiplying the outcome by a large factor $a$ is approximately $\log(a)$
times the treatment effect on the extensive margin,\footnote{We say $ f(a) = o(g(a)) $ if
$\lim_{a \to \infty} |f(a)/g(a)| = 0$. That is, as $a \to \infty$, $|f(a)|$ grows strictly
slower than $|g(a)|$.}
\begin{equation}
\E_P[ m(a Y(1) ) - m(a Y(0)) ] = \left( \P(Y(1) >0) - \P(Y(0)>0) \right)
\cdot \log(a) +  o(\log(a)). \label{eqn: log a rate main text}
\end{equation} 
\noindent Thus, the ATE for $m(Y)$ will tend to be more sensitive to finite changes in
scale the larger is the extensive margin treatment effect. This implies that sensitivity analyses that assess how treatment effect estimates for $m(Y)$ change under finite changes in the units of $Y$---or equivalently, under finite changes of $c$ in $\log(c+Y)$---are roughly equivalent to measuring the size of the extensive margin.

\end{rmk}

\begin{rmk}[Extension to continuous treatments] \label{rmk: nonbinary treatments}
 We focus on ATEs for binary treatments for expositional simplicity, although similar
 results apply with continuous treatments. In \Cref{subsec: continuous treatment}, we show
 that when $d$ is a continuous treatment, any treatment effect contrast that averages $m(aY(d))$
 across possible values of $d$ (i.e. a parameter of the form $\int \omega(d) E[m(aY(d)]$)
 is arbitrarily sensitive to scaling when there is an extensive margin effect.
\end{rmk}

\begin{rmk}[Extension to OLS estimands] \label{rmk: ols estimands}
It is worth noting that the results in this section show that population ATEs for $m(Y)$
are sensitive to the units of $Y$. These results are about \emph{estimands}, and thus any
consistent \emph{estimator} of the ATE for $m(Y)$ will be sensitive to scaling (at least
asymptotically). Thus, our results apply to ordinary least squares (OLS) estimators when
they have a causal interpretation, but also to non-linear estimators such as
inverse-probability weighting or doubly-robust methods. Nevertheless, given the prominence
of OLS in applied work, and the fact that OLS is sometimes used for non-causal estimands,
in \Cref{subsec: ols estimands} we provide a result specifically on the scale-sensitivity
of the population regression coefficient for a random variable of the form $m(Y)$ on an
arbitrary random variable $X$.
Our result
shows that the
coefficients on $X$ will be
arbitrarily sensitive to the scaling of $Y$ when the coefficients of a regression of
$\one[Y>0]$ on $X$ are non-zero. Thus, the OLS estimand using a logarithm-like function on
the left-hand side will be sensitive to scaling even when it does not have a causal
interpretation.
\end{rmk}

\begin{rmk}[Statistical significance] \label{rmk: tstats}
Equation~\eqref{eqn: log a rate main text} shows that $\P(Y(1) >0) -
\P(Y(0)>0)$ is the dominant term in $\theta(a)$ for large $a$, which suggests that the $t$-statistic for an estimator of $\theta(a)$ will generally converge to that for the analogous estimator of the extensive margin effect, $\P(Y(1) >0) -
\P(Y(0)>0)$. \Cref{prop: OLS SEs} in the appendix formalizes this intuition when the
treatment effects are estimated via OLS: As $a$ is made large, the $t$-statistic for
$\hat\theta(a)$ converges to that for the extensive margin estimate. In our empirical
analysis of papers in the \emph{American Economic Review} below, we find that indeed the
$t$-statistics for estimates of the ATE using $\arcsinh(Y)$ are typically close to those for the
extensive margin effect.
\end{rmk}

\begin{rmk}[When most values are large]
Researchers often have the intuition that if most of the values of the outcome are large,
then ATEs for transformations like $\log(1+Y)$ or $\arcsinh(Y)$ will approximate
elasticities, since $m(Y) \approx \log(Y)$ for most values of $Y$. Indeed, in an
influential paper, \citet{bellemare_elasticities_2020} recommend that researchers using
the $\arcsinh(Y)$ transformation should transform the units of their outcome so that most
of the non-zero values of $Y$ are large. The results in this section suggest---perhaps
somewhat counterintuitively---that if one rescales the outcome such that the non-zero
values are all large, the behavior of the ATE will be driven nearly entirely by the effect
of the treatment on the extensive margin and \textit{not} by the distribution of the
potential outcomes conditional on being positive. Moreover, the rescaling can be chosen to
generate any magnitude for the ATE if the treatment affects the extensive margin.
\end{rmk}

\begin{rmk}[Zero extensive margin] \label{rmk: zero extensive margin}
\cref{prop: can get any value for ATE} applies to settings where treatment has a non-zero
effect on average on the extensive margin. This raises the question of whether the use of
log-like transformations is justified in the absence of an extensive margin treatment effect. Our
\cref{prop: identified iff log} below implies that the ATE for any log-like transformation
will be sensitive to the units of the outcome for at least some distribution with
\emph{strictly} positive
outcomes, but perhaps not arbitrarily so in the sense of
\cref{prop: can get any value for ATE} (see \cref{subsec: no zeros discussion} for further
discussion). Moreover, if one were confident that the extensive margin effect were exactly
zero for all individuals, one could recover the ATE in logs for individuals with positive
outcomes by simply dropping individuals with $Y=0$. The use of log-like transformations is
thus somewhat difficult to justify even in settings without an extensive margin.
\end{rmk}

\subsection{Empirical illustrations from the \emph{American Economic Review}\label{subsec:
numerical illustration}}

We illustrate the results in this section by evaluating the sensitivity to scaling of
estimates using the $\arcsinh(Y)$ transformation in recent papers in the \emph{American
Economic Review} (\emph{AER}). In November 2022, we used Google Scholar to search for
``inverse hyperbolic sine'' among papers published in the \emph{AER} since 2018. We
searched for papers using $\arcsinh(Y)$ rather than $\log(1+Y)$ since the former are
easier to find with a simple keyword search. Our search returned 17 papers that estimate
treatment effects for an $\arcsinh$-transformed outcome.\footnote{We consider papers with
both binary and non-binary treatments, as our theoretical results extend easily to
non-binary treatments; see \Cref{rmk: nonbinary treatments}. Seven of the 10 papers we
replicated used a binary treatment.} Of these, 10 explicitly interpret the results as
percentage changes or elasticities, and 6 of the remaining 7 do not directly interpret the
units. See \Cref{tbl:selected-quotes-aer} for a list of the papers and relevant quotes. Of
the 17 total papers using $\arcsinh(Y)$, 10 had publicly available replication data that
allowed us to replicate the original estimates and assess their sensitivity to
scaling.\footnote{We include one paper where there was a slight discrepancy between our replication
of the original result and the result reported in the paper that only affected the
third decimal place.} For our replications, we focus on the first specification using
$\arcsinh(Y)$ presented in a table in the paper, which we view as a reasonable proxy for
the paper's main specification using $\arcsinh(Y)$.\footnote{We use the first coefficient
presented in a figure for one paper without any tables in the main text using
$\arcsinh(Y)$. If the first specification is a validation check (e.g. a pre-trends test),
we use the first specification of causal interest.}

We assess the sensitivity of these results by re-running exactly the same procedure as in
the original paper, except replacing $\arcsinh(Y)$ with $\arcsinh(100 \cdot Y)$. Thus, for
example, if the original paper estimated a treatment effect for the $\arcsinh$ of an
outcome measured in dollars, we use the same procedure to re-estimate the treatment effect
for the $\arcsinh$ of the outcome measured in cents. Since \eqref{eqn: log a rate main text} shows
that the sensitivity to scaling depends on the size of the extensive margin effect, we
also estimate the extensive margin effect by using the same procedure as in the original
paper but with the outcome $\one[Y>0]$.

\begin{table}[!ht]
    \captionsetup[table]{labelformat=empty,skip=1pt}
    \begin{tabular}{lrrrR{0.08\textwidth}R{0.08\textwidth}}
\toprule
 & \multicolumn{2}{c}{Treatment Effect Using:} &  & \multicolumn{2}{c}{\thead{Change from \\ rescaling units:}} \\ 
 \cmidrule(lr){2-3} \cmidrule(lr){5-6}
Paper & $\arcsinh(Y)$ & $\arcsinh(100 \cdot Y)$ & Ext. Margin & Raw & \% \\  
\midrule
Azoulay et al (2019) & 0.003 & 0.017 & 0.003 & 0.014 & 464 \\ 
Fetzer et al (2021) & -0.177 & -0.451 & -0.059 & -0.273 & 154 \\ 
Johnson (2020) & -0.179 & -0.448 & -0.057 & -0.269 & 150 \\ 
Carranza et al (2022) & 0.201 & 0.453 & 0.055 & 0.252 & 125 \\ 
Cao and Chen (2022) & 0.038 & 0.082 & 0.010 & 0.044 & 117 \\ 
Rogall (2021) & 1.248 & 2.150 & 0.195 & 0.902 & 72 \\ 
Moretti (2021) & 0.054 & 0.068 & 0.000 & 0.013 & 24 \\ 
Berkouwer and Dean (2022) & -0.498 & -0.478 & 0.010 & 0.020 & -4 \\ 
Arora et al (2021) & 0.113 & 0.110 & -0.001 & -0.003 & -3 \\ 
Hjort and Poulsen (2019) & 0.354 & 0.354 & 0.000 & 0.000 & 0 \\ 
\bottomrule
\end{tabular}

    \caption{Change in estimated treatment effects from re-scaling the outcome by a factor
    of 100 in papers published in the \emph{AER}  using $\arcsinh(Y)$ }
    \label{tbl:recale-by-100-aer}
    \floatfoot{Note: This table replicates treatment effect estimates using $\arcsinh(Y)$
    as the outcome in recent papers published in the \emph{AER}, and explores their sensitivity to the units of $Y$. The first column shows the author(s) and date of the paper. The second column shows the treatment effect on $\arcsinh(Y)$ using the units originally reported in the paper. The third column shows a treatment effect estimate constructed identically to the estimate in column 2 except using $\arcsinh(100 \cdot Y)$ as the outcome instead of $\arcsinh(Y)$, e.g. converting $Y$ from dollars to cents before taking the $\arcsinh$ transformation. The fourth column shows an estimate of the size of the extensive margin, obtained using $\one[Y>0]$ as the outcome. The final two columns show the raw difference and percentage difference between the second and third columns. The table is sorted on the magnitude of the percentage difference.} 
\end{table}

\begin{figure}[!htb]
    \centering
    \includegraphics[width = 0.7\linewidth]
    {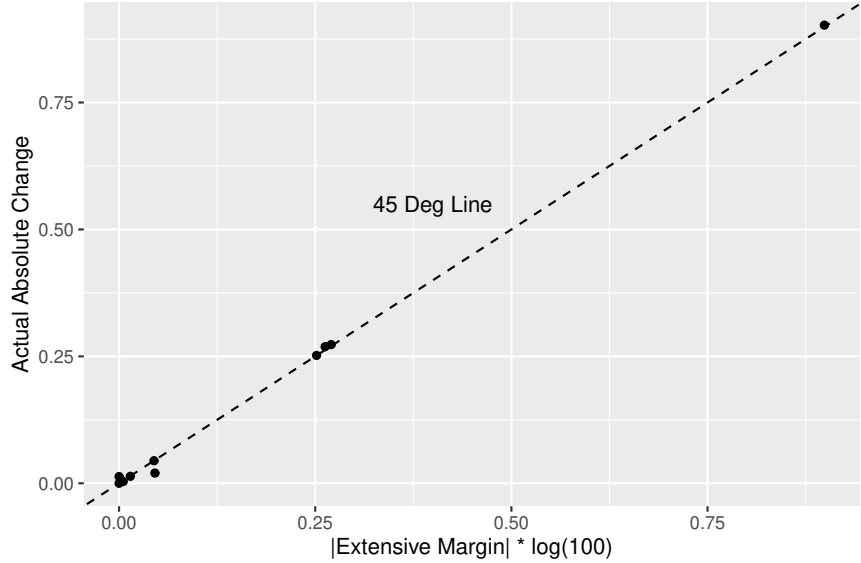}
    \caption{Change from multiplying outcome by 100 versus extensive margin effect}
    \label{fig:actual-change-versus-predicted-a100}
    \floatfoot{{Note:} This figure shows the relationship between the sensitivity of treatment effects using $\arcsinh(Y)$ to re-scaling the units of $Y$ and and the size of the extensive margin. For each replicated paper, this figure plots the absolute
    value of the change in the estimated treatment effect from multiplying the outcome by
    100 (i.e. the absolute value of the Raw Change column in \cref{tbl:recale-by-100-aer}) on the $y$-axis against $\log(100)$ times the absolute value of the extensive
    margin effect on the $x$-axis. If the approximation in \eqref{eqn: log a rate main
    text} were exact, all points would lie on the 45 degree line.}
\end{figure}

The results of this exercise, shown in \Cref{tbl:recale-by-100-aer}, illustrate that
treatment effect estimates can be quite sensitive to the scaling of the outcome when the
extensive margin is not approximately zero. Indeed, in 5 of the 10 replicable papers,
multiplying the
outcome by a factor of 100 changes the estimated treatment effect by more than 100\% of
the original estimate. The change in the estimated treatment effect is less than 10\% only in three papers, all of which have either zero or near-zero ($<$1 p.p.) effects on the extensive margin. \Cref{fig:actual-change-versus-predicted-a100} shows that the (absolute)
change in the estimated treatment effect is larger when the extensive margin effect is
larger, with the change lining up very closely with the approximation given in \eqref{eqn:
log a rate main text}.\footnote{In \Cref{fig:tstats}, we plot the $t$-statistics for the treatment effects estimates as well as those for the extensive margin effect. In line with the discussion in \Cref{rmk: tstats}, we find that the $t$-statistics for the treatment effect using $\arcsinh(Y)$ tend to be similar to those for the extensive margin, except when the extensive margin is very small, and become even closer when multiplying the units by 100.}

Using the same 10 papers, we also estimate treatment effects using $\log(1+Y)$ as the
outcome, and analogously explore how the results change when we multiply the units of $Y$
by 100. (Four of the 10 papers that we replicate report an alternative specification using $\log(1+Y)$ in the paper.) The results, shown in
\Cref{tbl:recale-by-100-aer-log1p}, are qualitatively quite similar those in
\Cref{tbl:recale-by-100-aer}, with five of the 10 treatment effect estimates again
changing by more than 100\%. These results underscore the fact that \cref{prop: can get
any value for ATE} applies to \emph{all} log-like transformations, including both
$\arcsinh(Y)$ and $\log(c+Y)$ for any constant $c$.

\section{Sensitivity to scaling for other ATEs}
\label{sec: trilemma}

Our results so far show that ATEs for transformations that are defined at zero and
approximate $\log(y)$ are arbitrarily sensitive to scaling. What other options are
available when there are zero-valued outcomes? To help delineate alternative options, in
this section we provide a result showing what properties a parameter defined with
zero-valued outcomes can have. Specifically, we establish a ``trilemma'': When there are
zero-valued outcomes, there is no parameter that is (a) an average of individual-level
treatment effects of the form $\theta_g = \E_P[ g(Y(1),Y(0)) ]$, (b) scale-invariant, and
(c) point-identified.\footnote{\label{fn:ate}Of
course, not
all parameters of the form $\E_P[g(Y(1), Y(0))]$ can be interpreted as an average of individual treatment
effects. For example $E[\one[Y(1)>0,Y(0)>0]]$ is the fraction of individuals whose outcomes is positive under both treatments, rather than a treatment effect. Our results apply to all parameters of this form, regardless of whether they are average treatment effects \emph{per se}.} Any approach for settings with zero-valued outcomes
must therefore abandon one of the properties (a)--(c); in \Cref{sec: recommendations} below
we discuss several approaches that relax one (or more) of these requirements.

Before stating our formal result, we must make precise what we mean by scale-invariance
and point-identification. We say that $g$ is scale-invariant if its value is the same
under any re-scaling of the units of $y$ by a positive constant $a$.

\begin{defn}
We say that the function $g$ is \textit{scale-invariant}
if it is homogeneous of degree zero, i.e. $g(y_1,y_0) = g(a y_1, a y_0)$ for all
$a,y_1,y_0>0$.
\end{defn}

 \Copy{jointdist}{We next describe point-identification. We consider parameters that are identified without placing restrictions on treatment effect heterogeneity. As in \citet{fan_partial_2017}, this is formalized by considering parameters that can be learned if we know the marginal distributions of $Y(1)$ and $Y(0)$, but not the full joint distribution of $(Y(1),Y(0))$. 

To connect treatment effect heterogeneity to the joint distribution of potential outcomes, consider the simple case of a randomized experiment. By examining the outcome distribution for the treated group, we can learn the marginal distribution of $Y(1)$. Likewise, by examining the outcome distribution for the control group, we can learn the marginal distribution of $Y(0)$. If treatment effects were assumed to be constant, then for each observed treated unit with outcome $Y(1)$, we could infer their untreated outcome as $Y(0) = Y(1) - \tau$, where $\tau$ is the average treatment effect. Hence, the joint distribution of $(Y(1),Y(0))$ would be identified. However, if we allow for treatment effect heterogeneity, then for an observed treated unit with outcome $Y(1)$, we do not know what their value of $Y(0)$ would be, and thus we do not know the joint distribution of $(Y(1),Y(0))$. This winds up being especially important in settings with an extensive margin, since when we observe the distribution of outcomes for treated units, it means that we do not know \emph{which} of the treated units would have had a zero outcome under the control condition, and thus it is difficult to disentangle the intensive and extensive margins.\footnote{In \cref{sec: structural equations}, we discuss a variety of structural approaches that impose assumptions restricting the joint distribution, thus allowing us to separately point-identify the effects for the two margins.} }%

With that intuition in mind, we now give a formal definition. Recall that $P$ denotes the
joint distribution of $(Y(1),Y(0))$, while $P_{Y(d)}$ denotes the marginal distribution of
$Y(d)$. We then say $\theta_g$ is point-identified if it depends on $P$ only through the
marginals $P_{Y(1)},P_{Y(0)}$.

\begin{defn}[Identification]
We say that $\theta_g$ is \emph{point-identified from the marginals at $P$} if for every
joint distribution $Q$ with the same marginals as $P$ (i.e. such that $Q_{Y(d)} =
P_{Y(d)}$ for $d=0,1$), $E_P[g(Y(1),Y(0))] = E_Q[g(Y(1),Y(0))]$. For a class of
distributions $\mathcal{P}$, we say that $\theta_g$ is \emph{point-identified over
$\mathcal{P}$} if for every $P \in
\mathcal{P},$ $\theta_g$ is point-identified from the marginals at $P$.
\end{defn}
\noindent We will denote by $\mathcal{P}_{+}$ the set of distributions on $[0,\infty)^2$.
Thus, $\theta_g$ is point-identified over $\mathcal{P}_{+}$ if it is always identified
when $Y$ takes on zero or weakly positive values. Our next result formalizes that it is
not possible to have a parameter of the form $E_P[g(Y(1),Y(0))]$ that is both
scale-invariant and point-identified over $\mathcal{P}_+$.

\begin{restatable}[A trilemma]{prop}{trilemma}
\label{prop: trilemma}
The following three properties cannot hold simultaneously:

\begin{enumerate}[label=(\alph*)]
    \item 
    $\theta_g = E_P[g(Y(1),Y(0))]$ for a non-constant function $g: [0,\infty)^2
    \rightarrow \reals$ that is weakly increasing in its first argument.
    
    \item
    The function $g$ is scale-invariant.
    
    \item
    $\theta_g$ is point-identified over $\mathcal P_{+}$.\footnote{A minor technical complication arises from the fact that $E_P[g(Y(1),Y(0)]$ could be infinite for some $P$. For the purposes of our result, it suffices to trivially define $\theta_g$ to be identified in this case. Alternatively, the same result holds if part (c) is modified to impose only that $\theta_g$ is point-identified over all distributions in $\mathcal{P}_+$ with finite support, thus avoiding issues related to undefined expectations.}
\end{enumerate}
\end{restatable}

\noindent Any parameter defined with zero-valued outcomes must therefore abandon one of properties (a)--(c).

As a special case, \Cref{prop: trilemma} implies that the ATE for any increasing function
$m(Y)$ defined at zero cannot be scale-invariant. This is because the ATE for $m(Y)$ takes
the form in (a) with $g(y_1,y_0) = m(y_1) - m(y_0)$, and is also point-identified (part
(c)). It follows that property (b) must be violated, i.e. there is some $c, y_0, y_1 > 0$
such that $m(cy_1) - m(cy_0) \neq m(y_1) - m(y_0)$. \cref{prop: trilemma} thus formalizes
the sense in which it is not possible to ``fix'' the issues with ATEs for log-like
transformations described above by taking alternative transformations of the outcome (e.g.
$\sqrt{Y}$).

\subsection{Implications for settings without an extensive margin\label{subsec: no zeros discussion}}

The trilemma in \cref{prop: trilemma} applies for transformations of the outcome defined at zero. To prove \Cref{prop: trilemma}, however, we establish an even stronger result: the only parameter satisfying properties (a) and (b) that is point-identified over distributions for which
$Y$ is \textit{strictly} positively-valued is the ATE in logs.\footnote{More precisely, the only such treatment effect is the ATE in logs or an affine tranformations thereof.} This result, which is formalized in \Cref{prop: identified iff log} in the Appendix, has some useful implications for settings in which the outcome is strictly positive. 

First, it implies that the ATE for any transformation of the outcome other than $\log(Y)$ will depend on the units of the outcome for at least some DGP where the outcome is strictly positive. The scale-dependence of log-like transformations such as $\log(1+Y)$ or $\arcsinh(Y)$ is thus not entirely limited to settings with an extensive margin.\footnote{There is thus no conflict between our results and those in \citet{thakral2023estimates}, who note that semi-elasticities for OLS regressions using log-like transformations may depend on the units of the outcome even when $Y$ is strictly positively-valued.} We note, however, that while the ATE for such transformations may depend on the units of the outcome even without zero-valued outcomes, the dependence need not be arbitrarily bad in the sense of \Cref{prop: can get any value for ATE}. Indeed, \eqref{eqn: log a rate main text} shows that if there is no extensive margin, the ATE for a log-like transformation will be approximately insensitive to scaling once the values of $Y$ are made large. This is intuitive, since if $Y$ is strictly positively-valued, the ATE for a log-like transformation will be approximately equal to the ATE in logs when the values of $Y$ are made large. 

Second, \cref{prop: identified iff log} implies that even when $Y(1)$ and $Y(0)$ are
strictly-positively valued, the average proportional effect $\theta_{\text{Avg}\%} =
E[(Y(1)-Y(0))/Y(0)]$ is not point-identified. This parameter is empirically relevant: For
instance, \citet{andrews_optimal_2013} show that in the
\citet{baily_aspects_1978}--\citet{chetty_general_2006} model with heterogeneous
consumption responses to unemployment, the optimal level of unemployment insurance depends
on a parameter of the form $\theta_{\text{Avg}\%}$, where $Y$ is consumption and $D$ is
unemployment. Although the ATE in logs may approximate $\theta_{\text{Avg}\%}$ when the
proportional effect of the treatment is approximately constant, our results imply that it
is not possible to point-identify $\theta_{\text{Avg}\%}$ when allowing for arbitrarily heterogeneous
proportional effects.

\section{Empirical approaches with zero-valued outcomes\label{sec: recommendations}}

Our theoretical results above imply that when there are zero-valued outcomes, the
researcher should not take a log-like transformation of the outcome and interpret the resulting
ATE as an average percentage effect: Unlike a percentage, such an ATE depends on the units of
the outcome. In this section, we highlight some other parameters that are well-defined and
easily interpreted when there are zero-valued outcomes; in \cref{sec: empirical} below, we show how these parameters can be estimated in three empirical applications. Of course, any alternative
parameter must necessarily drop one of the requirements in the trilemma in \Cref{prop:
trilemma}, but the choice of which to drop may depend on the researcher's motivation.

To inform our discussion of alternative parameters, it is therefore useful to first
enumerate several reasons why empirical researchers may target treatment
effects for a log-transformed outcome rather than the ATE in levels:
\begin{enumerate}[label=(\roman*),wide]

\item The researcher is interested in reporting a treatment effect parameter with
easily-interpretable units, such as ``percentage changes.''

\item The researcher believes that there are decreasing returns to the outcome, and thus
wants to place more weight on treatment effects for individuals with low initial outcomes.
For instance, the researcher may perceive it to be more meaningful to raise income from
$Y(0) = \$\text{10,000}$ to $Y(1) = \$\text{20,000}$ than from $Y(0) = \$\text{100,000}$
to $Y(1) = \$\text{110,000}$, yet both of these treatment effects contribute equally to
the ATE in levels. 

\item The researcher is interested in both the intensive and extensive margin effects of
the treatment, and is using the ATE for a log-like transformation as an approximation to
the proportional effect along the intensive margin.
\end{enumerate}

These three motivations suggest different ways of breaking out of the trilemma in
\Cref{prop: trilemma}. If the goal is to achieve a percentage interpretation, then one can
consider scale-invariant parameters outside of the class $E_P[g(Y(1),Y(0))]$. For
instance, researchers can consider the
ATE in levels expressed as a percentage of the control mean, or the ATE for a normalized
parameter $\tilde{Y}$ that already has a percentage interpretation. Alternatively, if the
goal is to capture concave social preferences over the outcome, then it is natural to
specify how much we value the intensive margin relative to the extensive margin---thus
abandoning scale-invariance. Finally, if the goal is to separately understand the
intensive margin effect, the researcher can abandon point-identification (from the
marginal distributions) and directly target the partially identified parameter $\E\bk{\log
(Y(1)) - \log(Y(0))
\mid Y(0) > 0, Y(1) > 0}$, the effect in logs for individuals with positive outcomes under
both treatments. We address each of these cases in turn below, with a summary of some possible parameters in
\cref{tab:rec_summary}. 

\begin{table}[htb]
  \caption{Summary of alternative target parameters}
  \label{tab:rec_summary}
  \centering
  \footnotesize
  \begin{tabular}{p{0.2\textwidth}p{0.3\textwidth}p{0.15\textwidth}p{0.35\textwidth}}
  \toprule
Description & Parameter    & \thead{Main property \\ sacrificed?}  & Pros/Cons
\\\midrule
Normalized ATE & $\E[Y(1)-Y(0)]/\E[Y(0)]$    & $E[g(Y(1),Y(0))]$         
&
\makecell{\textit{Pro:} Percent interpretation  \\ \textit{Con:} Does not capture decreasing returns}
\\  \\
Normalized outcome & $\E[Y(1)/X - Y(0)/X]$  & $E[g(Y(1),Y(0))]$          
& \makecell{\textit{Pro:} Per-unit-$X$ interpretation           
\\ \textit{Con:} Need to find sensible $X$} \\ \\
\makecell{Explicit tradeoff of \\ intensive/extensive \\ margins} & ATE for $m(y) =\begin{cases} \log(y) & y>0 \\ -x & y=0\end{cases}$ &
Scale-invariance & \makecell{\textit{Pro:} Explicit tradeoff of two margins 
\\ \textit{Con:} Need to choose $x$; Monotone only if \\ \hspace{.65cm} support excludes
$(0, e^{-x})$ \\} \\
Intensive margin effect & $\E\bk{\log \left(\frac{Y(1)}{Y(0)}\right) \mid Y(1) > 0, Y(0) > 0}$    &
Point-identification
&
\makecell{\textit{Pro:} ATE in logs for the intensive margin \\ \textit{Con:} Partial identification}\\\bottomrule
  \end{tabular}
\end{table}

\begin{rmk}[Statistical reasons for transforming the outcome]
We focus on settings where the researcher is interested in a parameter other than the ATE
in levels. \Copy{estimatedifficult}{In some settings, the researcher may be interested in
the ATE in levels, but simple regression estimators may be noisy owing to a long
right-tail of the outcome \citep{athey2021semiparametric}.} The researcher might then try
to estimate the ATE in levels by first estimating the ATE for a log-like transformation,
and then multiplying by the baseline mean. However, since the ATE for a log-like
transformation depends on the units of the outcome---and is thus not a true ``percentage''
effect---the validity of this approach for recovering the ATE in levels will depend on the
initial units of $Y$.\footnote{Even in the case where $Y$ is strictly positive and one
first estimates the ATE in logs, this approach will only recover the ATE in levels under
certain homogeneity assumptions, e.g. constant proportional effects. See
\citet{wooldridge_alternatives_1992} for related discussion.\label{fn:wooldridge92}} We
refer the reader to \citet{athey2021semiparametric} and \citet{muller_more_2023} for alternative
approaches to estimation and inference targeted to settings where the ATE in levels is of
interest but the outcome has heavy tails.
\end{rmk}
\Copy{identificationdid}{
\begin{rmk}[Transformation-specific identification] 
\label{rmk: identification did}
Another reason that researchers may consider taking a transformation of the outcome is
that a parametric assumption used for identification may be more plausible for some
functional forms than others. For example, when the outcome is strictly positive, parallel
trends in logs may be more plausible than parallel trends in levels if time-varying
factors are thought to have a multiplicative impact on the outcome. We note that
justifying parallel trends for a log-like transformation is especially tricky, however,
since if parallel trends holds for the $\arcsinh$ of an outcome measured in dollars, say,
it will not generally hold for the $\arcsinh$ of the outcome measured in cents
\citep{roth_when_2023}. Thus, the parallel trends assumption is specific to both the
transformation $m(\cdot)$ \textit{and} the units of the outcome. Moreover, even if the
researcher is confident in parallel trends for a particular log-like transformation and
unit of the outcome, our results imply that they should not interpret the resulting ATT as
an average percentage effect, since that ATT is dependent on the units in which the
outcome is measured (\Cref{prop: can get any value for ATE}).

In what follows, we consider alternative parameters that may be of interest when the
marginal distributions of the potential outcomes are identified for some population of
interest. Such identification is obtained in RCTs or under conditional unconfoundedness
(for the full population), as well in instrumental variables settings (for the population
of compliers), as these designs do not rely on functional form assumptions for
identification. If the original identification strategy relies on a functional form
assumption (e.g. parallel trends), then obtaining identification of the alternative
parameters discussed below may require different identifying assumptions. We discuss these
issues in detail in \cref{subsec: sequeira}, where we revisit the
difference-in-differences application in \citet{sequeira_corruption_2016}.
\end{rmk}
}

\subsection{When the goal is interpretable units \label{subsec: interpretable units}}

We first consider the case where the researcher's primary goal is to obtain a treatment
effect parameter with easily interpretable units, such as percentages.

\paragraph{Normalizing the ATE in levels.} 
One possibility is to target the parameter
\[\thetaPoisson = \frac{\E[ Y(1) - Y(0) ]}{\E[Y(0)]},\] 
\noindent which is the ATE \textit{in levels} expressed as a \textit{percentage of the 
control mean}. For example, if a researcher is studying a program $D$ meant to reduce
healthcare spending $Y$, then $\thetaPoisson$ is the percentage reduction in
costs from implementing the program.  This parameter is point-identified and
scale-invariant, and thus has an intuitive percentage interpretation. Importantly, however, $\thetaPoisson$ is the percentage change in the average outcome between treatment and control, but is \textit{not} an average of individual-level percentage changes.\footnote{This is roughly analogous to how quantile treatment effects show changes in the quantiles of the potential outcomes distributions, but \textit{not} the quantiles of the treatment effects (without further assumptions).} That is, $\thetaPoisson$ does not take the form
$E_P[g(Y(1),Y(0))]$, thus avoiding the trilemma in \Cref{prop: trilemma}. 

We note that
$\thetaPoisson$ is consistently estimable by Poisson regression (see \citet{gourieroux1984pseudo}; \citet{silva_log_2006}; \citet[][Chapter 18.2] {wooldridge2010econometric}) under an
appropriate identifying assumption. 
With a randomly assigned $D$, for example, estimation of $Y = \exp(\alpha + \beta D) U$ by
Poisson quasi-maximum likelihood (QMLE) consistently estimates the population coefficient
$\beta$, which  satisfies $e^{\beta}-1 = \E[Y(1)] / \E[Y(0)] -1 =
\thetaPoisson $. In \cref{sec: empirical} below, we illustrate how $\theta_{\text{ATE}\%}$ can be estimated by Poisson regression in practice in several empirical examples, including both an RCT and DiD setting.

\Copy{nuisancezeros}{We also emphasize that $\thetaPoisson$ is influenced by treatment effects
along both the intensive and extensive margins. In particular, the numerator of
$\thetaPoisson$ is the ATE in levels. Thus, if an individual has a treatment effect of say
1, that contributes the same to $\thetaPoisson$ regardless of whether their outcome
changes from 0 to 1 (an extensive margin change) or 1 to 2 (an intensive margin change).
The parameter $\thetaPoisson$ may therefore be attractive in settings where the researcher
does not want to distinguish between the intensive and extensive margins. For example, if
$Y$ is a count of publications by a researcher in a particular year, and publications are
sometimes zero owing to the idiosyncracies of the publication process, then it may be
reasonable to view a change between 0 and 1 as similar to a change between 1 and 2. On the
other hand, in settings where a zero corresponds to a distinct economic choice, such as
not participating in the labor market, then it may be of interest to separate the effects
along the intensive and extensive margin, as we discuss in more detail in \cref{subsec:
two margins} below.}

\Copy{elonmusk}{It is also worth noting that if the researcher has determined that the ATE
in levels is not of economic interest, then similar issues will likely arise for
$\thetaPoisson$, since $\thetaPoisson$ is just a re-scaling of the ATE in levels. For one,
the ATE in levels (and hence $\thetaPoisson$) imposes no diminishing returns, and thus
might be dominated by individuals in the tail of the outcome distribution, particularly
when the outcome is skewed. Whether this is warranted will depend on the economic
question: if the policy-maker's goal is to reduce healthcare spending, it may not matter
whether the savings are produced mainly by reducing spending for a small fraction of individuals with
catastrophic medical spending. On the other hand, a policy that increases every American's
income by \$100 and one that increases Elon Musk's income by \$35 billion and has no
effect on anyone else would have approximately the same value of $\thetaPoisson$, yet the
former may be vastly preferred by an inequality-minded policy-maker.} We therefore next
turn to alternative approaches that place less weight on the tails of the outcome
distribution.

\paragraph{Normalizing other functionals.} While $\thetaPoisson$ normalizes
the ATE by the control mean, one can obtain scale-invariance by normalizing other
functionals of the potential outcomes distributions.\footnote{Indeed,
any functional $\phi(P)$ is homogeneous of degree zero if and only if it can be written as
the ratio of two homogeneous of degree one
functionals.} For example, $$ \theta_{\text{Median}\%} = \dfrac{\mathrm{Median}(Y(1)) -
\mathrm{Median}(Y(0))}{\mathrm{Median}(Y(0))},$$ is the quantile treatment effect at the median normalized by
the median of $Y(0)$.\footnote{Note that $\theta_{\text{Median}\%}$ is well-defined only if $\mathrm{Median}(Y(0))>0$.\label{fn:median}} Put otherwise, it captures the percentage change in the median between the treated and control distributions. ($\theta_{\text{Median}\%}$ thus may be particularly relevant for politicians interested in maximizing the happiness of the median voter!) As is typically the case
with quantile treatment effects, however, the numerator of $\theta_{\text{Median}\%}$ need
not correspond to the median of individual-level treatment effects. Moreover, in many
settings, decision-makers may care about treatment effects throughout the distribution,
not just at the median, in which case $\theta_{\text{Median}\%}$ may not be the most economically-relevant parameter.

\paragraph{Normalizing the outcome.} A related approach to obtaining a treatment
effect with more intuitive units is to estimate the ATE for a transformed outcome that has
a percentage interpretation. One example is to consider an outcome of the form $\tilde{Y}
= Y/X$, where $Y$ is the original outcome and $X$ is some pre-determined characteristic.
For example, suppose $Y$ is employment in a particular area. The treatment effect in
levels for $Y$ may be difficult to interpret, since a change in employment of 1,000 means
something very different in New York City versus a small rural town. However, if $X$ is
the area's population, then $\tilde{Y}$ is the employment-to-population ratio, which may
be more comparable across places, and is already in percentage (i.e. per capita) units. We
note that the ATE for $\tilde{Y}$ is a scale-invariant, point-identified parameter of the
form $\theta = E_P[g(Y(1),Y(0),X)]$, and thus escapes the trilemma in \Cref{prop:
trilemma} by avoiding property (a).\footnote{It is scale-invariant in the sense that
$g(y_1,y_0,x) = g(ay_1,ay_0,a x)$.} The viability of this approach, of course, depends on
having a variable $X$ such that the normalized outcome $\tilde{Y}$ is of economic
interest. We suspect that in many contexts, reasonable options will be available,
including pre-treatment observations of the outcome (assuming these are positive), or the
\textit{predicted} control outcome given some observable characteristics (i.e., $X =
E[Y(0)
\mid W]$,
for observable characteristics $W$). 

A second example is to use $\tilde{Y} = F_{Y^*}(Y)$, where $F_{Y^*}$ is the cumulative distribution function (CDF) of some
reference random variable $Y^*$, as suggested in \citet{delius_cash_2020}. The transformed outcome $\tilde{Y}$ then corresponds to
the rank (i.e. percentile) of an individual in the reference distribution, and the ATE for
$\tilde{Y}$ can be interpreted as the average change in rank caused by the
treatment. The ATE for $\tilde{Y}$ is unit-invariant so long as $Y$ and $Y^*$ and measured
in the same units. Outcomes of this form have become increasingly popular in the
literature on intergenerational mobility, where $\tilde{Y}$ corresponds to a child's rank
in the national income distribution. This approach has been found to yield more stable
estimates than approaches using $\log(c+Y)$, which \citet{chetty_where_2014} show are
sensitive to the choice of $c$.\footnote{Similar to the discussion in
\cref{fn:wooldridge92}, the treatment effect in ranks cannot be converted back to obtain
the ATE in levels without additional assumptions.\label{fn:convertbackige}}

Finally, the researcher might report treatment effects on transformed outcomes of the form
$1[Y \geq y]$ for different values of $y$. For example, the researcher might report the
impact of the treatment on the probability that an individual earns at least \$50,000,
\$60,000, etc., and interpret it as the treatment effect on the probability of obtaining a
``well-paying job.''\footnote{The researcher could also report the implied CDF of $Y(1)$
and $Y(0)$, from which one can infer the treatment effect on outcomes of this form for all
$y$.} Such treatment effects have interpretable units as {percentage points} (i.e.
changes in probabilities). We note that treatment effects for outcomes of this form
combine the effect of the treatment along the intensive and extensive margin, since for
example, a worker who has $Y(1) > \$50,000 > Y(0)$ could either not work under control
($Y(0)=0$) or work under control but have earnings below \$50,000.

\subsection{When the goal is to capture decreasing returns \label{subsec: m0x}}

We next consider the case where the researcher wants to
capture some form of decreasing marginal utility over the outcome. For example, when $Y$
is strictly positively valued, the ATE in logs corresponds to the change in utility from
implementing the treatment for a utilitarian social planner with log utility over
the outcome, $U = E[\log(Y)]$. Intuitively, this social welfare function captures the fact
that the planner values a percentage point of change in the outcome equally for all
individuals, regardless of their initial level of the outcome. 

Of course, log utility is not well-defined when there is an extensive margin: A coherent utility function defined with zero-valued outcomes must take a stand on the relative importance of the intensive versus extensive margins. Recall from \cref{sub:any_value_ate} that when using
transformations like $\log(1+y)$ or $\arcsinh(y)$, the scaling of the outcome implicitly
determines the weights placed on these margins. 

Instead of implicitly weighting the margins via the scaling of $Y$, a more transparent
approach is to explicitly take a stand on how much one values the two margins of
treatment. Of course, if one knows that their utility is captured by $U = E[m(Y)]$ (for a
particular unit of $Y$, say earnings in dollars), then the ATE for $m(Y)$ is appropriate.
If one is unsure exactly of their utility function, then a rough calibration is to specify
how much one values a change in earnings from 0 to 1 relative to a percentage change in
earnings for those with non-zero earnings. If, for example, one values the extensive
margin effect of moving from 0 to 1 the same as a $100x$ percent increase in earnings,
then one might consider setting $m(y) =
\log(y)$ for $y>0$ and $m(0) = -x$. The ATE for this transformation can be interpreted as
an approximate percentage (log point) effect, where an increase from 0 to 1 is valued at
$100x$ log points.\footnote{Note that this transformation will generally only be sensible
if
the support of $Y$ excludes $(0, e^{-x})$, since otherwise the function $m(y)$ is not
monotone in $y$ over the support of $Y$. It is common, however, to have a lower-bound on
non-zero values of the outcome; e.g., a firm cannot have between 0 and 1 employees. In our application to \citet{sequeira_corruption_2016} below, we normalize the minimum non-zero value of $Y$ to 1 when applying this approach.\label{fn: normalize y for m0}}

We emphasize that for a fixed value of $x$, this approach necessarily depends on the
scaling of the outcome (thus avoiding the trilemma in \Cref{prop: trilemma}). However,
this may not be so concerning since the appropriate choice of $x$ also depends on the
units of the outcome---e.g., saying a change from 0 to 1 is worth $100x$ percent means
something very different if 1 corresponds with one dollar versus a million dollars. In
other words, ATEs for transformations such as $\arcsinh(Y)$ may be difficult to interpret
because the scaling of the outcome implicitly determines the relative importance of the
intensive and extensive margins; this approach avoids that difficulty by \emph{explicitly}
taking a stand on the tradeoff between these two margins. Nevertheless, a challenge with
this approach is that researchers may have differing opinions over the appropriate choice
of $x$ (or more generally, over the appropriate utility function).

\subsection{When the goal is to understand intensive and extensive margins \label{subsec: two margins}}

Finally, we consider the case where the researcher is interested in understanding the
intensive and extensive margin effects separately. A common question in the literature on
job training programs \citep{card_active_2010}, for instance, is whether a program raises
participants' earnings by helping them find a job---which would be expected only to have
an extensive-margin effect---or by increasing human capital, which would be expected to
also affect the intensive margin. In such settings, it is natural to target separate
parameters for the intensive and extensive margins. 

For example, the parameter
$$\theta_{\mathrm{Intensive}} = \E[
\log(Y(1)) - \log(Y(0)) \mid Y(1) >0, Y(0) >0]$$ captures the ATE in logs for those who
would have a positive outcome regardless of their treatment status. The parameter
$\theta_{\mathrm{Intensive}}$ is scale-invariant but is not point-identified from the
marginal distributions of the potential outcomes (thus avoiding the trilemma in \Cref{prop:
trilemma}), and therefore cannot be consistently estimated without further
assumptions.\footnote{$\theta_{\mathrm{Intensive}}$ also does not take the form
$E_P[g(Y(1),Y(0))]$, although it can be written as $$\dfrac{E_P\bk{ \one[Y(1)>0,Y(0)>0]
\log(Y(1)/Y(0)) } }{ E_P[ \one[Y(1)>0,Y(0)>0]]}, $$ where both the numerator and
denominator take this form.} However, \citet{lee_training_2009} popularized a method for
obtaining bounds on $\theta_{\mathrm{Intensive}}$ under the monotonicity assumption that,
for example, everyone with positive earnings without receiving a training would also have
positive earnings when receiving the training.\footnote{See, also,
\citet{zhang_estimation_2003} for related results, including bounds without the
monotonicity assumption.} Bounds on $\theta_{\mathrm{Intensive}}$ can be reported
alongside measures of the extensive margin effect, such as the change in the probability
of having a non-zero outcome, $P(Y(1)>0) - P(Y(0)>0)$. One can also potentially tighten
the bounds (or restore point-identification) by imposing additional assumptions on the
joint distribution of the potential outcomes---we provide an example of this in our
application to \citet{carranza2022job} below; see \citet{zhang_evaluating_2008, zhang_likelihood-based_2009} for
related approaches.\footnote{We note that the \citet{lee_training_2009} bounds will tend
to be tight when the extensive margin effect is close to zero. As noted in \Cref{rmk:
finite changes in scale}, this is precisely the setting where ATEs for log-like
transformations are relatively insensitive to finite changes in scale.}

We note that the parameter $\theta_{\mathrm{Intensive}}$ is generally distinct from the
``intensive margin'' marginal effects implied by two-part models (2PMs), which were
recommended for scenarios with zero-valued outcomes by \citet{mullahy_why_2023}, among
others. In \Cref{sec: 2pms}, we consider the causal interpretation of the marginal effects
of 2PMs, building on the discussion in \citet{angrist_estimation_2001}. Our decomposition
shows that the marginal effects from 2PMs yield the sum of a causal parameter similar to
$\theta_{\mathrm{Intensive}}$ as well as a ``selection term'' comparing potential outcomes
for individuals for whom treatment only has an intensive margin effect to those with an
extensive margin effect. It thus will generally be difficult to ascribe a causal interpretation to
the marginal effects of 2PMs without assumptions about this selection.

\section{Empirical applications}
\label{sec: empirical}

In this section, we focus on three concrete empirical
applications to illustrate how the alternative parameters described in \cref{sec: recommendations} can be
estimated in practice. To illustrate a range of possible applications, we consider a randomized controlled trial, 
a difference-in-differences design, and an instrumental variables design. 

\subsection{An RCT setting: \citet{carranza2022job}}

\citet{carranza2022job} conduct a randomized controlled trial (RCT) in South Africa. Individuals randomized to the treatment group are provided with certified test results that they can show to prospective employers to vouch for their skills. Individuals in the control group do not receive test results.\footnote{Some individuals are also assigned to a ``placebo'' arm in which they are provided the test results but the form does not include the individual's name, and thus cannot credibly be shared with employers. We focus on the effect of the main treatment relative to the pure control group.} They then investigate how this treatment impacts labor market outcomes such as employment, hours worked, and earnings. We focus here on the effects on hours worked. 

\paragraph{Original specification and sensitivity to units.} \citet{carranza2022job} estimate the effect of their randomized treatment on the inverse hyperbolic sine of weekly hours worked. Formally, they estimate the OLS regression specification 
\begin{equation} \arcsinh(Y_i) = \beta_0 + D_i \beta_1 + X_i'\gamma + u_i, \label{eqn: carranza regression} \end{equation}
\noindent where $Y_i$ is average weekly hours worked for unit $i$, $D_i$ is an indicator
for whether unit $i$ was in the treatment group, and $X_i$ is a vector of
controls.\footnote{\citet{carranza2022job} include individuals
receiving the ``placebo'' treatment in the sample and add an indicator for
receiving the placebo treatment in $X_i$. We follow the same practice, although the
results are similar if units receiving the placebo treatment are dropped.} Their estimate
of the ATE ($\hat\beta_1$) is 0.201 (see column (1) in \Cref{tbl: carranza original}).
They interpret this as a 20\% change in hours: ``Certification increases average weekly
hours worked, coded as zero for nonemployed candidates, by 20 percent'' (p. 3560).

\begin{table}[ht!]
\centering
\begin{tabular}[t]{lccc}
\toprule
  & arcsinh(weekly hrs) & arcsinh(yearly hrs) & arcsinh(FTEs)\\
\midrule
Treatment & 0.201 & 0.417 & 0.031\\
 & (0.052) & (0.096) & (0.012)\\
Units of outcome: & Weekly Hrs & Yearly Hrs & FTEs\\
\bottomrule
\end{tabular}

\caption{Estimates using $\arcsinh(Y)$ with different units of $Y$ in \citet{carranza2022job}}
\label{tbl: carranza original}
\floatfoot{Note: This table shows estimates of the average treatment effect in \citet{carranza2022job} on the inverse hyperbolic sine of hours worked, estimated using \eqref{eqn: carranza regression}. In the first column, the outcome is the inverse hyperbolic sine of \emph{weekly} hours, as in the original paper. The remaining columns use the inverse hyperbolic sine of annualized hours (weekly hours times 52) or the inverse hyperbolic sine of the number of full-time equivalents worked (weekly hours divided by 40). Standard errors are clustered at the assessment date (the unit of treatment assignment) as in the original paper.}
\end{table}

However, the results in \Cref{sec: loglike results} suggest that the estimate of $\beta_1$
should not be interpreted as a percentage effect, since it depends on the units of the
outcome. To illustrate this, in columns (2) and (3) we re-estimate specification \eqref{eqn: carranza
regression} with $Y_i$ redefined to be (a) yearly hours worked, i.e. weekly hours times 52, or
(b) the number of full-time equivalents (FTE) worked, i.e. weekly hours divided by 40. The
results change quite substantially depending on the units used, with an estimate of 0.417
 using yearly hours and 0.031 using FTEs. We therefore turn next to alternative approaches with a percentage interpretation in this setting.

\paragraph{Percentage changes in the average.} The average number of (weekly) hours worked
was 9.84 in the treated group and 8.85 in the control group. A simple summary of the
treatment effect is thus that average hours worked were 11\% higher in the treated group
($9.84/8.85=1.11$). This is an estimate of the parameter $\thetaPoisson = E[Y(1)-Y(0)]/E[Y
(0)]$ discussed in \Cref{subsec: interpretable units} above. A numerically equivalent way to obtain this estimate of 11\% is to use Poisson quasi-maximum likelihood estimation (Poisson QMLE) to estimate
\begin{equation} Y_i = \exp(\beta_0 + \beta_1
D_i) U_i  \label{eqn: poisson carranza} \end{equation}  and then calculate
$\thetaPoissonHat = \exp(\hat\beta_1) - 1 = 0.11$ (see column (1) in \Cref{tbl: poisson
carranza}).\footnote{This estimation is done in the sample of treated units and control
units, discarding the placebo group. One could equivalently retain the units in the
placebo group and add an indicator for the placebo group to \eqref{eqn: poisson carranza}.}
This formulation in terms of Poisson QMLE is useful since it allows us to include
covariates to potentially increase precision. Column (2) of \Cref{tbl: poisson carranza}
shows the
estimate of $\thetaPoissonHat$ from estimating
\begin{equation} Y_i = \exp(\beta_0 + \beta_1
D_i + X_i'\gamma) U_i  \label{eqn: poisson carranza w covs} \end{equation} 
by Poisson QMLE, with smaller standard errors than in
column (1) (0.069 vs. 0.081).

\begin{table}[!ht]
\centering
\begin{tabular}[t]{lcc}
\toprule
  & (1) & (2)\\
\midrule
$\beta_0$ & 2.180 & 0.150\\
 & (0.058) & (0.311)\\
$\beta_1$ & 0.106 & 0.150\\
 & (0.072) & (0.060)\\
Implied Prop. Effect & 0.112 & 0.150\\
 & (0.081) & (0.069)\\
\midrule
Covariates & N & Y\\
\bottomrule
\end{tabular}

\caption{Poisson Regression and Implied Proportional Effects in \citet{carranza2022job}.\label{tbl: poisson carranza}}
\floatfoot{Note: the first two rows of column (1) show the estimates of the coefficients $\beta_0$ and $\beta_1$ in \eqref{eqn: poisson carranza}, estimated using Poisson QMLE. The third row shows the implied estimate of the proportional effect, $E[Y(1)-Y(0)]/E[Y(0)]$, calculated as $\thetaPoissonHat=\exp(\hat\beta_1) - 1$. The second column shows analogous estimates using \eqref{eqn: poisson carranza w covs}, which adds controls for pre-treatment covariates (we do not show the coefficients on the controls in the interest of brevity). Standard errors are clustered at the assessment date (the unit of treatment assignment) as in the original paper.}
\end{table}

\paragraph{Separate estimates for the extensive/intensive margins.} As shown in \Cref{tbl:recale-by-100-aer}, the treatment in \citet{carranza2022job} has an estimated extensive margin treatment effect of 0.055, meaning that it increases the fraction of people with positive hours worked by 5.5 percentage points. We may be interested
in whether the overall 11\% increase in hours worked is driven entirely by the extensive
margin, or whether there is an intensive margin effect. That is, does the treatment increase hours only by bringing people into the labor force, or does it also allow people who would have worked anyway to find jobs with more hours (e.g. full-time instead of part-time)? To this end, we can use the method of
\citet{lee_training_2009} to compute bounds for the effect of the treatment for
``always-takers'' who would have positive hours worked regardless of treatment
($Y(1)>0,Y(0)>0$).\footnote{We again exclude units receiving the
``placebo treatment.''} The Lee bounds approach requires the monotonicity assumption that
anyone who would work positive hours without the treatment would also work positive hours
when treated (i.e., $\P(Y(1) = 0, Y(0)>0) = 0$). This seems reasonable if
workers only share the
information provided by the
treatment when it helps their job prospects. It could be violated, however, if workers
mistakenly share their test score results when in fact employers view them negatively.

Column 1 of \Cref{tbl:lee-carranza} reports bounds of $[-0.20,0.28]$ for the effect of the
treatment on log hours worked by the always-takers, while Column 2 shows bounds of
$[-6.67,2.77]$ for weekly hours (in levels). Unfortunately, in this setting the Lee bounds
are fairly wide, including both a zero intensive-margin effect as well as fairly large
intensive-margin effects (up to 28 log points). Thus, without further assumptions, the
data is not particularly informative about the size of the intensive margin.

We can, however, say more if we are willing to impose some assumptions about how the
always-takers, who would work regardless of treatment status ($Y(1) > 0, Y(0) > 0$),
compare to the compliers ($Y(1) > 0, Y(0) = 0$), who
only work positive hours when receiving the treatment. We might reasonably expect that the
compliers are negatively selected relative to the always-takers and thus would work fewer
hours when receiving treatment. We can formalize this by imposing that $E[Y(1) \mid
\text{Complier}] = (1-c) E[Y(1) \mid \text{Always-taker}]$, i.e. that average hours worked
for compliers under treatment is $100 c$\% lower than for always takers. Columns 3 through
5 of \Cref{tbl:lee-carranza} report estimates of the average effect on the always-takers,
assuming $c = 0, 0.25$, and $0.5$, respectively.\footnote{Under the assumptions in \citet{lee_training_2009}, $E[Y(1) \mid Y(1)>0] = \theta E[Y(1) \mid \text{Always-taker}] + (1-\theta) E[Y(1) \mid \text{Complier}]$, where $\theta = P(Y(0)>0)/P(Y(1)>0)$. Plugging in $E[Y(1) \mid
\text{Complier}] = (1-c) E[Y(1) \mid \text{Always-taker}]$, it follows that $E[Y(1) \mid \text{Always-taker}] = 1/(\theta + (1-c)(1-\theta)) E[Y(1) \mid Y(1)>0]$. Further, $E[Y(0) \mid \text{Always-taker}] = E[Y(0) \mid Y(0)>0]$. Our estimation plugs in sample analogs to these expressions to estimate $E[Y(1)-Y(0) \mid \text{Always-taker}]$.} If we assume that always-takers and
compliers work an equal number of hours under treatment ($c=0$), then our point estimates suggest that there is actually a
negative intensive-margin effect for the always-takers ($-1.02$ weekly hours).
Under the assumption that compliers work 25\% fewer hours ($c=0.75$), the estimated effect
for always-takers is near zero ($-0.07$ weekly hours), consistent with no important
intensive margin. Finally, if we assume compliers work half as many hours as the
always-takers ($c=0.5$), then our estimates suggest a positive intensive margin
effect ($0.95$ weekly hours). Our assessment of the importance of the intensive margin
thus depends on how negatively-selected we think compliers are relative to always-takers.

\begin{table}[!ht]
    \centering
    \begin{tabular}[t]{lddddd}
\toprule
  & {(1)} & {(2)} & {(3)} & {(4)} & {(5)}\\
\midrule
Lower bound & -0.195 & -6.665 & {} & {} & {}\\
 & (0.064) & (1.366) & {} & {} & {}\\
Upper bound & 0.283 & 2.771 & {} & {} & {}\\
 & (0.114) & (2.067) & {} & {} & {}\\
Point estimate & {} & {} & -1.025 & -0.069 & 0.954\\
 & {} & {} & (1.182) & (1.349) & (1.588)\\
\midrule
units & {Log(Hours)} & {Hours} & {Hours} & {Hours} & {Hours}\\
$c$ & {} & {} & 0 & 0.25 & 0.5\\
\bottomrule
\end{tabular}

    \caption{Bounds and point estimates for the intensive margin treatment effect in \citet{carranza2022job}}
    \label{tbl:lee-carranza}
    \floatfoot{Note: This table shows bounds and point estimates of the intensive margin
    treatment effect in \citet{carranza2022job}, i.e. the treatment effect on hours worked
    for ``always-takers'' who would work positive hours regardless of treatment status.
    The first first two columns of the table show \citet{lee_training_2009} bounds for the
    effect of treatment on the always-takers when the outcome is $\log(\text{Hours})$ and
    weekly hours, respectively. Columns 3 through 5 show point estimates for the effect on
    weekly hours worked for always-takers under the assumption that average hours worked
    by ``compliers'' (who work only when treated) are $100c\%$ lower than for the
    always-takers. Standard errors are calculated via a non-parametric  bootstrap using
    1,000 draws, clustered at the assessment date level.}
\end{table}

\subsection{A DiD setting: \citet{sequeira_corruption_2016}\label{subsec: sequeira}}

\citet{sequeira_corruption_2016} studies a decrease in tariffs on trade between Mozambique and South Africa which occurred in 2008. She is interested in whether the reduction in tariffs reduced bribes paid to customs officers (among other outcomes). To study this question, she utilizes a difference-in-differences design comparing the change in bribes paid for products that were affected by the tariff change to that for a comparison group of products that did not experience a change in tariffs.

\paragraph{Original specification and sensitivity to units.} \citet{sequeira_corruption_2016} has repeated cross-sectional data with information on the bribe amount $Y_{it}$ paid on shipment $i$ in year $t$. She estimates the regression specification

\begin{equation}
 \log(1+Y_{it}) = \beta_0 + D_i \times \text{Post}_t \, \beta_1 + D_i \, \beta_2 + \text{Post}_t \, \beta_3 + X_{it}' \beta_4 + \epsilon_{it},   \label{eqn: sequeira original}
\end{equation}

\noindent where $D_i$ is an indicator for whether shipment $i$ is for a product type
affected by the tariff change in 2008, $\text{Post}_t$ is an indicator for whether year
$t$ is after the tariff change, and $X_{it}$ is a vector of covariates related to shipment
$i$ in period $t$. \citet{sequeira_corruption_2016} estimates \eqref{eqn: sequeira
original} with $Y_{it}$ measured in 2007 Mozambican Metical (MZN) and obtains
$\hat\beta_{1, (\text{MZN})} = -3.7$ (SE $=1.1$). However, estimating the same
specification with $Y_ {it}$ measured in thousands of U.S. dollars instead yields an estimate of $\hat\beta_{1, (\$1000)} = -0.11$ (SE $=0.070$).\footnote{We use
the conversion rate of $ 1 \text{ USD} = 24.48 \text{ MZN}$ as of January 1, 2007, as
provided by fxtop.com.} These results reinforce the conclusion from \Cref{sec: loglike
results} that treatment effects for $m(y) = \log(1+y)$ should not be interpreted as
approximating a percentage effect.

In what follows, we discuss a variety of alternative approaches that may be reasonable in
this context. We note that in a non-experimental setting like this, different approaches
may rely on different identifying assumptions. We therefore explicitly discuss the
identifying assumptions needed by each of the methods we discuss.

\paragraph{Proportional treatment effects.} One natural approach here is to
target the average proportional treatment effect on the treated,
$$\thetaPoissonATT = \dfrac{E[Y_{it}(1) \mid D_i =1, \text{Post}_t = 1] -
E[Y_{it}(0) \mid D_i =1, \text{Post}_t = 1]}{ E[Y_{it}(0) \mid D_i =1, \text{Post}_t = 1]
}.$$
\noindent This is the percentage change in the average outcome for the treated group in
the
post-treatment period.

Identification of $\thetaPoissonATT$ requires us to infer the counterfactual
post-treatment mean outcome for the treated group, $E[Y_{it}(0) \mid D_i =1, \text{Post}_t
= 1]$. Of course, one approach to obtain such identification would be to assume parallel
trends in levels. However, given that the treated and control groups have different
pre-treatment means (see the bottom panel of \Cref{tbl:sequeira merged}), it may be
unreasonable to expect
that time-varying factors (e.g. the macroeconomy) have equal level effects on the
outcome. An alternative identifying assumption is to impose that, in the absence of
treatment, the \emph{percentage} changes in the mean would have been the same for the
treated and control group.  As in \citet{wooldridge_simple_2023}, this can be formalized
using a ``ratio'' version of the parallel trends assumption,
\begin{equation}
\dfrac{E[Y_{it}(0) \mid D_i =1, \text{Post}_t = 1]}{E[Y_{it}(0) \mid D_i =1, \text{Post}_t
= 0]} = \dfrac{E[Y_{it}(0) \mid D_i =0, \text{Post}_t = 1]}{E[Y_{it}(0) \mid D_i =0,
\text{Post}_t = 0]} .    \label{eqn: ratio pt}
\end{equation}
\noindent Intuitively, \eqref{eqn: ratio pt} states that if the treatment had not
occurred, the average percentage change in the mean outcome for the treated group would
have been the same as the average percentage change in the mean outcome for the control
group. Under \eqref{eqn: ratio pt}, we can thus estimate the counterfactual percentage change in the mean outcome for the treated group using the observed percentage change for the control group.

\Cref{tbl:sequeira merged} shows that the sample mean of the outcome for the treated group
decreased by 75\% between the pre-treatment and post-treatment periods (from 4,742 to
1,172 (MZN)). Under the ratio parallel trends assumption
\eqref{eqn:
ratio pt}, this suggests that the mean outcome for the treated group would also have decreased
by 75\% in the absence of treatment, thus implying an estimate of $2,602$ for the counterfactual mean outcome for the treated group. The actual post-treatment mean for the treated group is 465, which is 82\% below this implied
counterfactual. This implies that the tariff reduction reduced the average bribe
in the post-treatment period by 82\%, i.e. $\thetaPoissonATTHat = -0.82$. Conveniently,
this estimate can also be obtained using Poisson QMLE to estimate
\begin{equation}
 Y_{it} = \exp( \beta_0 + D_i \times \text{Post}_t \, \beta_1 + D_i \, \beta_2 + \text{Post}_t \,
 \beta_3) \epsilon_{it}     \label{eqn: poisson sequeira}
\end{equation}

\noindent and then computing $\thetaPoissonATTHat = \exp(\hat\beta_1) - 1= -0.82$, as shown in column (1) of \Cref{tbl:sequeira merged}.

\begin{table}[!ht]
\begin{tabular}[t]{ldd}
\toprule
  & {(1)} & {(2)}\\
\midrule
Post $\times$ Treatment & -1.722 & -1.272\\
 & (0.632) & (0.606)\\
Prop. Effect & -0.821 & -0.720\\
 & (0.113) & (0.170)\\
\midrule
Covariates & {N} & {Y}\\ 
\midrule
Treated Group Means (Pre, Post): & {10527} & {465}\\
Treated Group Means (Pre, Post): & {4742} & {1172}\\
\bottomrule
\end{tabular}

    \caption{Poisson regression estimates of $\theta_{\text{ATT}\%}$}
    \label{tbl:sequeira merged}
    \floatfoot{Note: this table shows Poisson regression estimates of \eqref{eqn: poisson
    sequeira} and \eqref{eqn: poisson sequeira w cov} in columns (1) and (2), respectively.
    The first row of the table shows the estimate $\hat\beta_1$. The second row shows
    $\exp(\hat\beta_1) - 1$, which is the implied estimate of the proportional treatment
    effect $\thetaPoissonATT$. The coefficients on control variables are omitted for
    brevity. Standard errors are clustered at the four-digit product code as in the
    original paper. The mean bribe amounts (in MZN) by treatment group and time period are
    displayed in the bottom panel. The pre-period refers to the year 2007, whereas the
    post-treatment period is an average over the years 2008, 2011, and 2012 (the three
    post-treatment years for which data is available).}
\end{table}

We can also re-incorporate the covariates $X_{it}$ by estimating  
\begin{equation}
 Y_{it} = \exp( \beta_0 + D_i \times \text{Post}_t \, \beta_1 + D_i \, \beta_2 +
 \text{Post}_t \, \beta_3 + \beta_4' X_{it}) \epsilon_{it},   \label{eqn: poisson sequeira
 w cov}
\end{equation}
\noindent which yields an estimate of $\thetaPoissonATT$ of $-0.72$, as
shown in the second column of \cref{tbl:sequeira merged}. As formalized in
 \citet{wooldridge_simple_2023}, this estimate will be a consistent estimate of
 $\thetaPoissonATT$ if \eqref{eqn: ratio pt} holds conditional on $X_{it}$, and the
 conditional expectation of $Y_{it}$ takes the functional form implied by 
 \eqref{eqn:
 poisson sequeira w cov} (assuming $\epsilon_{it}$ has mean 1 conditional on the
 covariates). The approach with covariates thus suggests that the tariff change reduced
 the average bribe for treated products by 72\% in the post-treatment period.

\citet{sequeira_corruption_2016}'s data only contains information on one year prior to
treatment (2007), and so in this context it is not possible to evaluate the plausibility
of \eqref{eqn: ratio pt} using periods prior to the policy change of interest. If multiple
pre-treatment periods were available, however, one could estimate a Poisson QMLE
event-study of the form
\begin{equation}
 Y_{it} = \exp\pr{ \lambda_t + 
 D_i \, \beta_2 + \sum_{r \neq -1} D_i \times [\text{RelativeTime}_t = r] \,
 \beta^
 {ES}_r} \epsilon_{it} ,   \label{eqn: poisson es}
\end{equation}

\noindent where $\text{RelativeTime}_t = t-2008$ is the time relative to the treatment
date. The event-study coefficients $\beta^{ES}_r$ for $r<0$ are analogous to
``pre-trends'' coefficients in typical difference-in-differences
event-studies, and are informative about whether the pre-treatment analogue to \eqref{eqn: ratio pt}
holds.\footnote{More precisely, the exponentiated coefficients
$\exp(\hat\beta_{r})-1$ correspond to the implied ``placebo''  proportional treatment
effects for periods before treatment. We recommend plotting the exponentiated coefficients
in event-studies, although we note that $\exp(\beta) - 1 \approx \beta$ for $\beta \approx
0$. As with typical tests for pre-trends, one should be cautious that a
failure to reject the null that the pre-treatment coefficients equal zero does not
necessarily imply that the identifying assumption is satisfied
\citep{kahn-lang_promise_2020, roth_pretest_2022}. One can (partially) address these
issues by applying sensitivity analysis tools for event-studies
\citep[e.g.][]{rambachan_more_2023} to estimates of \eqref{eqn: poisson es} to further
gauge the robustness of the findings to violations of the identifying assumptions. We also
refer the reader to \citet{wooldridge_simple_2023} for extensions of the Poisson
regression approach to settings with staggered treatment timing.}

\paragraph{Log effects with calibrated extensive margin value.} The analysis above
presented estimates of $\thetaPoissonATT$, the proportional change in the
\emph{average} bribe caused by the treatment. It is well-known that averages can be
heavily influenced by observations in the tail, especially when the outcome has a skewed
distribution, as is the case here (see \Cref{fig:sequeira-density}). One might argue that
a world in which most products receive medium-sized bribes is more corrupt than one in
which a very small fraction of products receive large bribes---even if they both produce
the same average bribe amount. This motivates studying the treatment effect on a concave
transformation of the outcome that is less heavily influenced by outcomes in the tail of
the distribution. As an illustration of this, we first normalize the outcome so that $1$
corresponds to the value of the minimum non-zero bribe in the data (that is, we divide by
$y_{\min} = \min_{Y_{it}>0} Y_{it} = 15.68 \text{ MZN}$). We then estimate the treatment
effect for the transformed outcome $m(Y)$, where $m(y) =
\log(y)$ for $y>0$ and $m(0) = -x$ for some choice of $x$, as described in \Cref{subsec:
m0x}. If $x$ is set to 0, then this estimates the treatment effect in logs where all zero
bribes are set to equal the smallest positive bribe in the data; this specification thus
``shuts off'' the extensive margin change between 0 and $y_{\min}$. If instead $x$ is set
to $0.1$, for example, then a change between $0$ and $y_{\min}$ is valued as the
equivalent of a 10 log point change along the intensive margin.

We estimate the treatment effect for these transformations using the analogue to
\eqref{eqn: sequeira original} that replaces $\log(1+Y_{it})$ with $m(Y_{it})$ on the
left-hand side.\footnote{As usual, identification of the treatment effect for $m(Y)$ using
difference-in-differences requires parallel trends for $m(Y(0))$. The identifying
assumption thus varies depending on the choice of $x$. The results in
\citet{roth_when_2023} imply that parallel trends will hold for all values of $x$ when a
parallel trends assumption is satisfied for the distribution of $Y(0)$. If more
pre-treatment periods were available, these identifying assumptions could be partially
evaluated using pre-trends tests. See \Cref{rmk: identification did} for additional
discussion of identification.} The results for $x \in \{0,0.1,1,3\}$ are shown in
\Cref{tbl:sequeira-log0-x}. Column (1) shows an effect of 249 log points ($\hat\beta_1 =
-2.49$) when we treat zero bribes as if they were equal to $y_{\min}$ (i.e. setting
$x=0$). The estimated treatment effect grows in magnitude as we place more value on the
extensive margin by increasing $x$. Interestingly, the original estimate in
\citet{sequeira_corruption_2016} of $-3.748$ using $\log(1+Y)$ is similar to what we
obtain when we value a change from 0 to $y_{\min}$ at 300 log points ($x=3$). The original
specification can thus be viewed as placing a rather large weight on the extensive margin.

\begin{figure}
    \centering
    \includegraphics[width = .75\linewidth]{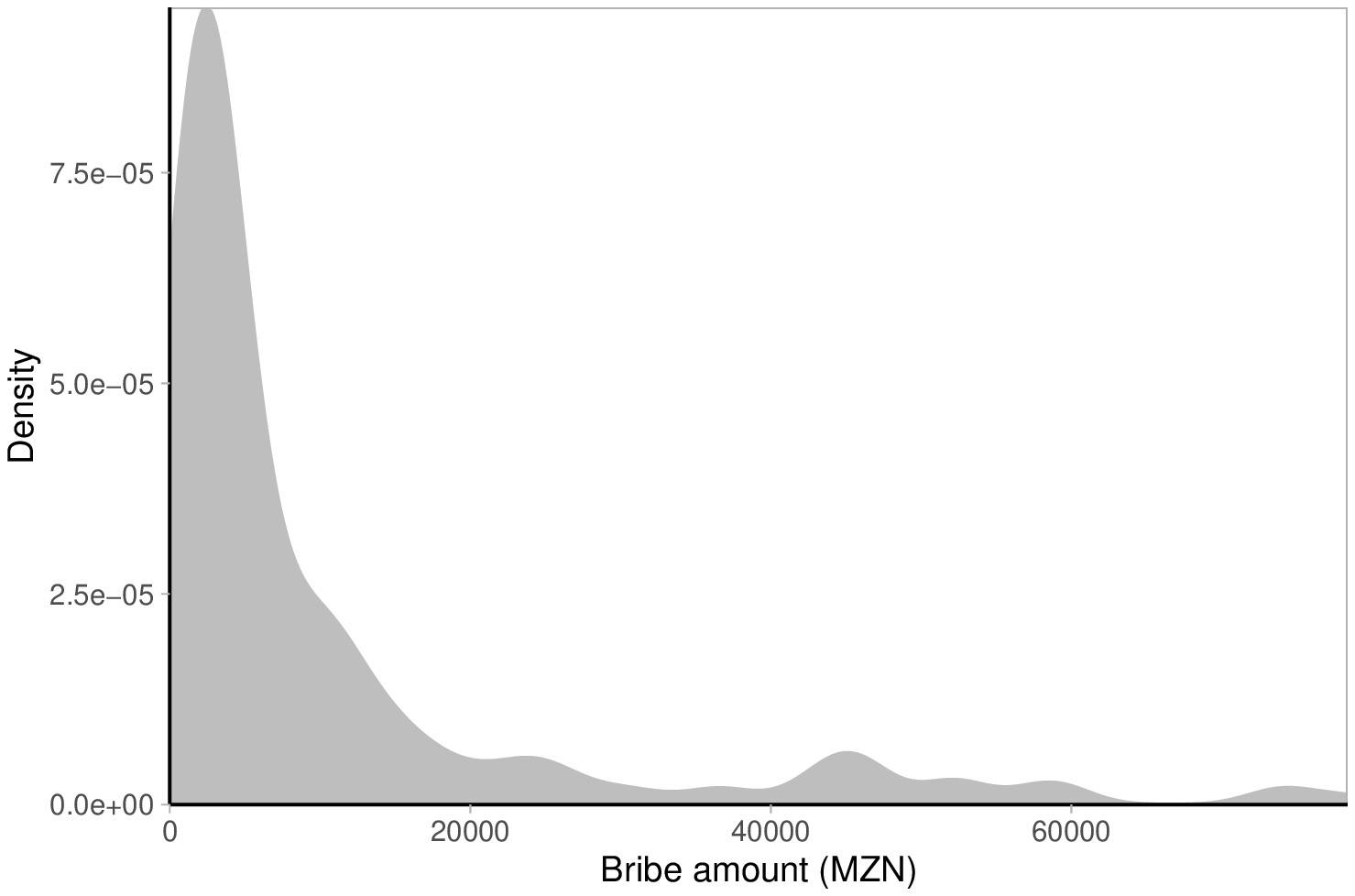}
    \caption{Density of bribe amount in \citet{sequeira_corruption_2016}}
    \label{fig:sequeira-density}
    \floatfoot{Note: this figure shows a kernel density estimate of the bribe amount in \citet{sequeira_corruption_2016}, pooling across all observations with a positive bribe. The kernel density estimates are constructed using the default settings of the \texttt{stat\_density} function in R.}
\end{figure}

\begin{table}[!ht]
    \caption{Explicit calibration of the extensive margin in \citet{sequeira_corruption_2016}}
    \label{tbl:sequeira-log0-x}
    \begin{tabular}[t]{ldddd}
\toprule
  & {(1)} & {(2)} & {(3)} & {(4)}\\
\midrule
Post $\times$ Treatment & -2.493 & -2.538 & -2.949 & -3.860\\
 & (0.740) & (0.752) & (0.861) & (1.106)\\
Extensive margin value ($x$): & 0.000 & 0.100 & 1.000 & 3.000\\
\bottomrule
\end{tabular}

    \floatfoot{Note: this table shows estimates of the treatment effect on the treated using $m(Y)$ as the outcome in \citet{sequeira_corruption_2016}, where $m(y)$ is defined to equal $\log(y)$ for $y>0$ and $-x$ for $y=0$. The outcome is normalized so that $Y=1$ corresponds to the minimum non-zero value of the outcome. Thus, the treatment effect assigns a value of $100 x$ log points to an extensive margin change between 0 and the minimum non-zero value of $Y$. The treatment effects are estimated using \eqref{eqn: sequeira original}, except replacing $\log(1+Y_{it})$ with $m(Y_{it})$. Standard errors are clustered at the four-digit product code as in the original paper.}
\end{table}

\subsection{An IV setting: \citet{berkouwer2022credit}}
\citet{berkouwer2022credit} conduct an RCT in Nairobi in which
they randomize the price for energy-efficient stoves. They use the randomized price 
($p_i$)
as an
instrument for whether an individual $i$ buys an energy-efficient stove ($D_i$). They use
this
instrument to estimate the effects of stove-adoption on outcomes such as charcoal spending
($Y_i$).

\paragraph{Original specification and sensitivity to scale.} Let $X_i$ be a vector of
control variables (including a constant). \citet{berkouwer2022credit} estimate
\begin{equation}\arcsinh(Y_{i}) = D_i \beta + X_i'\gamma + \epsilon_i \label{eqn: 2sls berkouwer}\end{equation}
\noindent by two-stage least squares (TSLS), using $p_i$ as an instrument for $D_i$.\footnote{More precisely, each observation $i$ is an individual-by-week pair, and some (but not all) individuals are surveyed on multiple weeks. Standard errors are clustered at the respondent level.} (They
also report results where spending is measured in levels.) The estimated coefficient
$\hat\beta$ is an estimate of the LATE of stove adoption on the $\arcsinh$ of charcoal
spending for instrument-compliers whose decision of whether to purchase the stove depends on the
price offered in the experiment.\footnote{We use the phrase ``instrument-compliers'' to distinguish compliers for the instrument, whose value of $D(z)$ depends on $z$, from ``compliers'' discussed earlier who have $Y(1)>0,Y(0)=0$. Since the instrument takes on multiple values
(i.e. multiple price offers), $\beta$ corresponds to a weighted average of treatment
effects across instrument-compliers for different values of the instrument
\citep{angrist_interpretation_2000}.}  In \citet{berkouwer2022credit}, $Y_i$ is measured
as weekly charcoal spending in dollars. They obtain a coefficient of $\hat\beta = -0.50$
and write ``[t]he 50 log point reduction corresponds to a 39 percent decrease in charcoal
consumption [since $\exp(-0.50)=1-0.39$]'' (p. 3306). 

However, if we change the units of the outcome to annual charcoal spending in Kenyan shillings, the original currency in which charcoal spending was measured, the same specification yields an estimate of
$-0.44$. Relative
to our previous applications, the change in the treatment effect estimates is fairly small
for these choices of units, due to a small estimated extensive margin of
0.01 (see \Cref{tbl:recale-by-100-aer}).\footnote{We note, however, that the
$t$-statistic for the effect on $\arcsinh(Y_i)$ is rather sensitive here, changing from
approximately 7 to 3 depending on the units.} Nevertheless, the fact that the treatment
effects using an $\arcsinh$-transformed outcome depend on the units should give us pause
in interpreting them as percentages. Indeed, a percentage effect is not well-defined for
someone who has non-zero spending under treatment and zero spending under the control, so
an average individual-level percentage effect does not make sense if the treatment can
affect whether one has any charcoal spending.

\citet{berkouwer2022credit} first discuss the LATE in levels, and then immediately
afterwards state that the treatment effect for the $\arcsinh$-transformed outcome
``corresponds to a 39 percent decrease in charcoal consumption'' (p. 3306). The main goal
of taking the $\arcsinh$ transformation here thus appears to be to obtain a treatment
effect with a percentage interpretation. We therefore next implement two approaches with
an (approximate) percentage interpretation in this context.

\paragraph{Proportional LATE.} One natural approach in this context is to estimate the
proportional change in the average outcome for instrument-compliers, i.e. to estimate $\thetaPoisson$
among the population of instrument-compliers. Put otherwise, we can express the LATE in levels as a
percentage of the control mean for instrument-compliers. An estimate of the LATE in levels is naturally
obtained using TSLS specification \eqref{eqn: 2sls berkouwer} with $Y_i$ as the outcome,
which yields an estimate of $-2.46$. As described in \citet{abadie_bootstrap_2002}, we can
likewise obtain an estimate of the control instrument-complier mean by using TSLS with $-(D_i-1)
\cdot Y_i$ as the outcome, which yields an estimate of $5.86$. Putting these together, we
obtain an estimate of $\thetaPoisson$ for instrument-compliers of $-2.46/5.86=-0.42$ (SE =
0.046), which suggests that average charcoal
spending is 42\% lower for instrument-compliers under treatment than under control.\footnote{The standard error was calculated via a non-parametric bootstrap with
1,000 draws, clustered at the respondent level. We note
that with a binary instrument, an estimate of $\theta_{\text{Intensive}}$ for instrument-compliers
can be obtained using Poisson IV regression (e.g. the ivpoisson command in Stata); see
\citet{angrist_estimation_2001}. However, we are not aware of a LATE interpretation of
Poisson IV regression with a multi-valued instrument, and thus do not pursue it here.
Whether Poisson IV regression has such an interpretation with a continuous IV strikes us
an interesting topic for future work.} If pollution is proportional to charcoal spending,
then this parameter is economically relevant as it corresponds to the percentage reduction
in pollution for instrument-compliers from gaining access to the efficient stove.

\paragraph{Lee bounds.} \citet{berkouwer2022credit} benchmark their treatment effect
estimates relative to engineering estimates of the efficiency gains of using an efficient
stove relative to a non-efficient one. For this benchmarking exercise, it seems sensible
to focus on the intensive-margin effect of the treatment---i.e., the treatment effect for
instrument-compliers who would use a non-efficient stove if offered a high price and an efficient one
if offered a low price. To do so, we can form \citet{lee_training_2009}-type
bounds for the average treatment effect in logs for instrument-compliers who would have positive
charcoal spending regardless of treatment status.\footnote{The validity of the
\citet{lee_training_2009}-type bounds requires the ``monotonicity'' assumption that all
instrument-compliers who would have some charcoal consumption when not buying an efficient stove
would also have some charcoal consumption when buying an efficient stove, which seems
reasonable. Note that this is a distinct assumption from the instrument
monotonicity assumption needed for a LATE interpretation for instrumental variables
\citep{imbens1994identification}, which in this context states that anyone who would buy a
stove at a higher price would also buy at a lower price.}

The bounds on $\theta_{\text{Intensive}}$ for instrument-compliers are $ [-0.565,-0.538]$ (with SEs
for the lower and upper bounds of 0.072 and 0.075).\footnote{We obtain these estimates
using the procedure in \citet{abadie_bootstrap_2002}, as described in detail in \Cref{sec:
berkouwer iv details}.} This implies that for the instrument-compliers who would spend on charcoal
regardless of treatment status, spending decreases by 54 to 56 log points. We note that
the Lee bounds are fairly tight in this case, as tends to be the case when the extensive
margin is small. It is also worth noting that in this example, the estimated treatment
effects using $\arcsinh(Y_i)$---both in terms of weekly spending in dollars and in terms
of annual spending in Kenyan shillings---fall outside of the Lee bounds, although they are
fairly close to the upper bound when using weekly spending in dollars.

\section{Conclusion}
It is common in empirical work to estimate ATEs for transformations such as $\log(1+Y)$ or
$\arcsinh(Y)$ which are well-defined at zero and behave like $\log(Y)$ for large values of
$Y$. We show that the ATEs for such transformations should not be interpreted as percentages,
since they depend arbitrarily on the units of the outcome when there is an extensive margin. Further, we show that any
parameter that is an average of individual-level treatment effects of the form $E_P[g(Y(1),Y(0))]$ must be scale-dependent if it is
point-identified and well-defined at zero. We discuss several alternative approaches,
including estimating scale-invariant normalized parameters (e.g. via Poisson regression), explicitly calibrating the value placed
on the intensive versus extensive margins, and separately estimating effects for the
intensive and extensive margins (e.g. using Lee bounds). We illustrate how these approaches can be applied in practice in three empirical applications.

\nocite{azoulay2019does,beerli2021abolition,berkouwer2022credit,cabral2022demand,carranza2022job,faber2019tourism,hjort2019arrival,johnson2020regulation,mirenda2022economic,norris2021effects,ager2021intergenerational,arora2021knowledge,bastos2018export,fetzer2021security,moretti2021effect,rogall2021mobilizing,cao2022rebel}

\newpage

{\singlespacing
\bibliographystyle{aer}
\bibliography{Bibliography}}

\newpage 
\appendix 
\begin{center}
{\Large\bfseries
Online Appendix for 

``Log with zeros? Some problems and solutions''
}

\vspace{1em}

\begin{tabular}[t]{c@{\extracolsep{2em}}c} 
\large{Jiafeng Chen} &  \large{Jonathan Roth}\\ 
\normalsize{\it Harvard Business School} &
\normalsize{\it Brown University} \\ 
\normalsize{\href{mailto:jiafengchen@g.harvard.edu}{jiafengchen@g.harvard.edu}} & \normalsize{\href{mailto:jonathan_roth@brown.edu}{jonathan\_roth@brown.edu}}
\end{tabular}
\vspace{2em}

\today
\end{center}

\DoToC
\setcounter{page}{0}
\thispagestyle{empty}

\newpage

\numberwithin{equation}{section}

\section{Proofs of results in the main text}

\subsection{Proof of \cref{prop: can get any value for ATE}}

\anyvalue*
\begin{proof}
Note that $\theta(0) = \E_P[m(0)] - \E_P[m(0)] = 0$. Additionally,  \cref{prop: log a
rate} below implies that $|\theta(a)| \to \infty$ as $a \to \infty$. To establish the
proof, it thus suffices to show that $\theta(a)$ is continuous on $[0,\infty)$. The
desired result is then immediate from the intermediate value theorem.

To establish continuity, fix some $a \in [0,\infty)$ and consider a sequence $a_n \to a$.
Without loss of generality, assume $a_n<a+1$ for all $n$. Let $m_{a_n}(y) = m(a_n y)$.
Since $m$ is continuous, $m_{a_n}(y) \to m_a(y)$ pointwise. 
We are done if we can apply the dominated convergence theorem to show that therefore $\E
[m_{a_n}(Y)] \to \E[m_a(Y)]$.

Since $m(y)/\log(y) \to 1$ as
$y \to \infty$, there exists $\bar{y}$ such that $m(y) < 2 \log(y)$ for all $y \geq
\bar{y}$. From the monotonicity of $m$, it follows that \begin{align*}
 m(0) \leq m(y) &\leq \one[y \leq
\bar{y}] m(\bar{y}) + \one[y > \bar{y}] 2\log(y)  \\
&\leq \eta + 2 \cdot \one[y > \bar{y}] \log(y), \numberthis \label{eqn: m bound}
\end{align*}
where $\eta = |m(\bar{y})|$, and hence
\begin{align*}
m(0) \leq m_{a_n}(y) &\leq \eta + 2 \cdot \one[a_n y>  \bar{y}]
 \log(a_n y) 
 \\ 
 &\leq \eta + 2 \cdot \one[y>0] \cdot (|\log(a+1)| + |\log(y)|) =:
 \bar{m}(y).
\end{align*}
\noindent for all $n$. Hence, we have that $|m_{a_n}(y)| \leq |m(0)| + \bar{m}(y)$ for
all $n$, and the bounding function is integrable for $Y(d)$ for $d=0,1$ by the fourth
assumption of the proposition. It follows from the dominated convergence theorem that
$\E_P[ m_{a_n}(Y(d))] \to \E_P[ m_a(Y(d)) ]$ for $d=0,1$, and thus $\theta(a_n) \to
\theta(a)$, as we wished to show.
\end{proof}

\subsection{Proof of \cref{prop: trilemma}}

\trilemma*
\begin{proof}
To establish the proof of \cref{prop: trilemma}, we rely on \cref{prop: identified iff
log}, which shows that the only scale-invariant parameter of the form $E_P[g(Y(1),Y(0))]$
that is identified over distributions on the positive reals is the ATE in logs (up to an
affine transformation).

Given \cref{prop: identified iff log}, note that if $g:[0,\infty)^2 \to \reals$ is
increasing in $y_1$, then it cannot be equal to $c \log(y_1/y_0) + d$ for $c >0$
everywhere on $(0,\infty)^2$, since this would imply that $\lim_{y_1 \to 0} g(y_1,1) =
-\infty < g(0,1)$. 
\Cref{prop: trilemma} is then immediate from \cref{prop:
identified iff log}, which shows that if properties (a) and (b) are satisfied, and
$\theta_g$ is point-identified over $\mathcal{P}_{++} \subset
\mathcal{P}_{+}$, then $g = c \log(y_1/y_0)+d$ on $(0,\infty)^2$. Thus, there does not
exist such a $g$.
\end{proof}

\begin{prop} 
\label{prop: identified iff log}
Let $\mathcal{P}_{++}$ denote the set of distributions over compact subsets of
$(0,\infty)^2$. Suppose $g: (0,\infty)^2 \rightarrow \reals$ is weakly increasing in $y_1$
and scale-invariant. Then $\theta_g$ is point-identified over $\mathcal{P}_{++}$ if and
only if $g(y_1,y_0) = c \cdot (\log(y_1) - \log(y_0)) + d$, for constants $c \geq 0$ and $d \in
\reals$.
\end{prop}

\begin{proof}
\label{appendix: proofs}
We first show that point-identification over $\mathcal{P}_{++}$ implies that $g(\cdot,\cdot)$ must be
additively separable. We do so by considering the points $\br{y_0, y_0 + b} \times
\br{y_1, y_1+a}$ on a rectangular grid. If $g(\cdot, \cdot)$ is not additively separable,
then its expectation with respect to distributions supported on the rectangular grid
depends on the correlation. Similar arguments appear in, e.g., \citet{fan_partial_2017}.

Formally, suppose that there there exist positive values $y_1,y_0,a,b > 0$ such that 
\[g(y_1,y_0) + g(y_1+a,y_0+b) \neq g(y_1+a,y_0) + g(y_1,y_0 + b).\] 
Now, consider the marginal distributions $\P_{Y(d)}$ such that $P(Y(1) = y_1) =
\frac{1}{2} = \P(Y(1) =y_1 + a)$ and $\P(Y(0) = y_0) = \frac{1}{2} = \P(Y(0) = y_0 + b)$.
Let $P_1$ and $P_2$ denote the joint distributions corresponding with these marginals and
perfect positive and negative correlation of the potential outcomes, respectively. Then we
have that
\begin{align*}
\E_{P_1}(g(Y(1),Y(0))) &= \frac{1}{2} \left( g(y_1,y_0) + g(y_1+a,y_0+b) \right) \\
& \neq
\frac{1}{2} \left( g(y_1+a,y_0) + g(y_1,y_0 + b) \right) \\&= \E_{P_2}(g(Y(1),Y(0))),
\end{align*}
and thus $\theta_g$ is not point-identified from the marginals at $P_1$. Hence, if
$\theta_g$ is identified over $\mathcal{P}_{++}$, then it must be that
\[g(y_1,y_0) + g(y_1+a,y_0+b) = g(y_1+a,y_0) + g(y_1,y_0 + b) \text{ for all } y_1,y_0,
a,b>0,\]
and hence
\[g(y_1+a,y_0)  - g(y_1,y_0)=  g(y_1+a,y_0+b) -g(y_1,y_0 + b) \text{ for all }
y_1,y_0,a,b>0.\] It follows that we can write $g(y_1,y_0) = r(y_1) + q(\frac{1}{y_0})$,
where $r(y_1) = g(y_1,1) - g(1,1)$ and $q(\frac{1}{y_0}) = g(1,y_0)$.

Second, we show that homogeneity of degree zero, combined with monotonicity, implies that
$g$ must be a difference in logarithms. 
Observe that since $g$ is scale-invariant, \[g(y_1,y_0) = g\left(\frac{y_1}{y_0}, \frac{y_0}
{y_0}\right) = g\left(\frac{y_1}{y_0}, 1\right) =: h\left(\frac{y_1}{y_0}\right),\]
\noindent where $h$ is an increasing function. We thus have that for any $a,b > 0$,
\begin{align*}
& g(1,1) = h(1) = r(1) + q(1) \\
& g(a,1) = h(a) = r(a) + q(1) \\
& g\pr{1,\frac{1}{b}} = h(b) = r(1) + q(b) \\
& g\pr{a,\frac{1}{b}} = h(ab) = r(a) + q(b) 
\end{align*}
and hence $h(ab) = h(a) + h(b) - h(1)$. It follows that $\tilde{h}(x) = h(x) - h(1)$ is an
increasing function such that $\tilde{h}(ab) = \tilde{h}(a) + \tilde{h}(b)$ for all $a,b
\in \reals$, i.e. an increasing function satisfying Cauchy's logarithmic function
equation: $\phi(ab) = \phi(a)+\phi(b)$ for all positive reals $a,b$. Recall that if a
function is increasing, then it has countably many discontinuity points, and thus is
continuous somewhere. It is a well-known result in functional equations that the only
solutions to Cauchy's logarithmic equation are of the form $\phi(t) = c \log(t)$, if we
require that these solutions are continuous at some point; see
\citet{aczel_lectures_1966}, Theorem 2 in Section 2.1.2.\footnote{Correspondingly,
non-trivial solutions to Cauchy's logarithmic equations are highly ill-behaved.} Since we
require monotonicity, the constant $c \geq 0$. Thus, $g(y_1, y_0) = h(y_1/y_0) = \tilde
h(y_1/y_0) + \tilde h(1) = c \log(y_1) - c \log (y_0) + \tilde h(1)$. Letting $d = \tilde
h(1)$ completes the proof of \Cref{prop: identified iff log}.
\end{proof}

\section{Extensions}

\subsection{Sensitivity to finite changes in scale}

The following result formalizes the discussion in \Cref{rmk: finite changes in scale}
about how the ATE for $m(Y)$ changes with finite changes in the scale of $Y$.

\begin{prop} \label{prop: log a rate}
Suppose that:
\begin{enumerate}
  \item $m: [0, \infty) \to \R$ is a weakly increasing function.
  \item $m(y) / \log(y) \to 1$ as $y \to \infty$.
  \item $\E_{P_{Y(d)}}[|\log Y(d)| \mid Y(d) > 0] < \infty$ for $d = 0, 1$. 
\end{enumerate}
Then,  as $a \to \infty$,
\[ \E_P[ m(a \cdot Y(1) ) - m(a \cdot Y(0)) ] = \left( \P(Y(1) >0) - \P(Y(0)>0) \right)
\cdot \log(a) +  o(\log(a)). \]
\end{prop}

\begin{proof}
Fix a sequence $a_n \to \infty$, and without loss of generality, assume $a_n > e$.  We
will show that
\begin{equation}
\frac{1}{\log a_n}\E_P[m(a_nY(1)) - m(a_nY(0))] \to  P(Y(1) = 0) - P(Y
(0) = 0). \label{eqn: desired ratio convergence}    
\end{equation}

\noindent Define $f_n(y) = m(a_n y) / \log(a_n)$. Note that $f_n(y) \to \one[y>0]$ pointwise, since $f_n(0) = m(0) / \log(a_n) \to 0$, while for $y>0$, 

$$f_n(y) = \frac{ m(a_ny) }{ \log(a_n) } = \frac{m(a_n y)}{ \log(a_n y) } \frac{ \log(a_n) + \log(y) }{ \log(a_n) } \to 1 ,$$

\noindent where we use the fact that $m(y)/\log(y) \to 1$ as $y\to \infty$ by assumption.
We apply the dominated convergence theorem to show that $\E_P [f_n(Y(d))] \to \P(Y
(d) >
0)$.

We showed in the proof to \Cref{prop: can get any value for ATE} that $$|m(y)| \leq \kappa
+ 2 \cdot \one[y>0] \cdot |\log(y)|$$
\noindent where $\kappa$ is a constant not depending on $y$.\footnote{In particular, 
\eqref{eqn: m bound} implies the inequality for $\kappa = \eta + |m(0)|$.} It follows that
$f_n$ is similarly dominated:
$$|f_n(y)| = \frac{|m(a_n y)| }{ \log(a_n) } \leq \kappa + 2 \cdot \one[y>0] \cdot (1 +
|\log(y)|).$$
\noindent Further, since $E_P[|\log(Y(d)) |\mid Y(d) >0]$ is finite by assumption, the
upper bound is integrable for $y= Y(d)$ for $d=0,1$. It follows from the dominated
convergence theorem that $$E_P[ f_n(Y(d)) ] = E_P\left[ \frac{m(a_nY(d))}{ \log(a_n) }
\right] \to E_P[ \one[Y(d)>0] ] = P(Y(d)>0).$$
\noindent Equation \eqref{eqn: desired ratio convergence} then follows from applying this result for
$d=0,1$ and taking the difference of the limits.
\end{proof}

\subsection{Extension to continuous treatments\label{subsec: continuous treatment}}

Although we focus on binary treatment in the main text for simplicity, similar issues arise with
continuously distributed $D$. Suppose now that $D$ can take a continuum of values on some
set $\mathcal{D}\subseteq \reals$. Let $Y(d)$ denote the potential outcome at the dose
$d$, and $P$ the distribution of $Y(\cdot)$. Consider the parameter
$$\theta(a) = \int_{\mathcal{D}} \omega(d) E_P[ m(a Y(d) ) ] ,$$

\noindent which is a weighted sum of the average values of $m(aY(d))$ across different
values of $d$ with weights $\omega(d)$. For example, in an RCT with a continuous
treatment, a regression of $m(a Y)$ on $D$ yields a parameter of the form $\theta(a)$
where, by the Frisch--Waugh--Lovell theorem, the weights are proportional to $(d - E[D])
p(d)$ and integrate to 0.\footnote{Here, $p(d)$ denotes the density of $D$ at $d$ over the
randomization distribution.}

We now show that $\theta(a)$ can be made to have arbitrary magnitude via the choice of $a$
when there is an extensive margin effect. In particular, by an extensive margin effect we
mean that $\int \omega(d) P(Y(d) >0) \neq 0$, i.e. when there is an average effect on the
probability of a zero outcome, using the same weights $\omega(d)$ that are used for
$\theta(a)$. When $\theta(a)$ is the regression of $m(aY)$ on $D$ in an RCT, for example,
$\int \omega(d) P(Y(d) >0) \neq 0$ if the regression of $\one[Y>0]$ on $D$ yields a
non-zero coefficient.

\begin{prop}
Suppose that: 

\begin{enumerate}
    \item The function
    $m$ satisfies parts 1 and 2 of \Cref{prop: can get any value for ATE}.

    \item
    (Extensive margin effect) $\int_{\mathcal{D}} \omega(d) P(Y(d) >0) \neq 0$.
    
    \item
    (Bounded expectations) For all $d$, $E_P[ | \log(Y(d)) | \mid Y(d) > 0] < \infty$.
    
    \item
    (Regularity for weights) The weights $\omega(d)$ satisfy $\int_{\mathcal{D}} \omega(d)
    =0$, $\int_{\mathcal{D}} |\omega(d)| < \infty $ and $\int_{\mathcal{D}} |\omega(d)|
    \cdot E_P[ | \log(Y(d)) | \mid Y(d) > 0] < \infty $.
    
\end{enumerate}
Then for every $\theta^* \in (0,\infty)$, there exists $a>0$ such $|\theta(a)|=\theta^*$. In
particular, $\theta(a)$ is continuous and $\theta(a) \to 0$ as $a \to 0$ and $|\theta(a)|
\to \infty$ as $a \to \infty$.

\end{prop}

\begin{proof}
Note that $\theta(0) = \int \omega(d) m(0) = 0$. It thus suffices to show that $\theta(a)$
is continuous for $a \in [0,\infty)$ and that $|\theta(a)| \to \infty$ as $a \to \infty$.
The result then follows from the intermediate value theorem.

We first show continuity. Fix $a \in [0,\infty)$ and a sequence $a_n \to a$. Let $f_n(d) =
\omega(d) E_P[ m(a_n Y(d)) ]$. We showed in the proof to \Cref{prop: can get any value for
ATE} that $E_P[ m(a_n Y(d)) ] \to E_P[ m(a Y(d)) ]$, and thus $f_n(d) \to \omega(d) E_P[ m(a
Y(d)) ]$ pointwise. We also showed in the proof to \Cref{prop: can get any value for ATE}
that for $a_n$ sufficiently close to $a$, $$|m(a_n Y)| \leq \kappa + 2 \cdot \one[y>0]
\cdot |\log(y)| ,$$
\noindent for a constant $\kappa$ not depending on $n$. It follows that 
$$|f_n(d)| \leq |\omega(d)| \cdot |\kappa| + 2|\omega(d)| \cdot E_P[|\log(Y(d))| \mid
Y(d)>0] , $$

\noindent and the upper bound is integrable by part 4 of the Proposition. Hence, by the
dominated convergence theorem, we have that $\theta(a_n) = \int_{\mathcal{D}} f_n(d) \to
\int_{\mathcal{D}} \omega(d) E_P[m(aY(d)] = \theta(a)$, as needed.

To show that $|\theta(a)| \to \infty$ as $a \to \infty$, we will show that $$
\frac{\theta(a)}{\log(a)}\to \int_{\mathcal{D}} \omega(d) P[Y(d) > 0] $$
\noindent as $a\to \infty$. Consider $a_n \to \infty$, and suppose without loss of
generality that $a_n>e$. Observe that $$\frac{\theta(a_n)}{\log(a_n)} = \int_{\mathcal{D}}
\omega(d) \dfrac{ E_P[m(a_n Y(d))] }{ \log(a_n) } .$$

\noindent We showed in the proof to \Cref{prop: log a rate} that for each $d$,

$$\dfrac{ E_P[m(a_n Y(d))] }{ \log(a_n) }  \to P(Y(d)>0) .$$

\noindent Letting $f_n(d) = \omega(d) \dfrac{ E_P[m(a_n Y(d))] }{ \log(a_n) }$, we thus have
that $f_n(d) \to \omega(d) P(Y(d)>0)$ pointwise. Moreover, we showed in the proof to
\Cref{prop: can get any value for ATE} that $$|m(y)| \leq \kappa + 2 \cdot \one[y>0] \cdot
|\log(y)|$$

\noindent where $\kappa$ is a constant not depending on $y$. It follows that 

$$\frac{|m(a_n y)| }{ \log(a_n) } \leq \kappa + 2 \cdot \one[y>0] \cdot (1 + |\log(y)|)$$

\noindent and thus that $$|f_n(d)| \leq |\omega(d)| \cdot (\kappa + 2+2 E_P[ |\log(Y(d)|
\mid Y(d) >0 ]) $$ where the upper bound is integrable by the fourth part of the
proposition. The result then follows from dominated convergence.

\end{proof}

\subsection{Extension to OLS estimands and standard errors \label{subsec: ols estimands}}

As noted in \Cref{rmk: ols estimands}, our results imply that any consistent estimator of
the ATE for an outcome of the form $m(aY)$ will be (asymptotically) sensitive to scaling
when there is an extensive margin effect. Our results thus cover the OLS estimator when it is consistent for the ATE for some (sub)-population $P$ (e.g. in an RCT or under unconfoundedness). Given the prominence of OLS in
applied work---and the fact that it is sometimes used for non-causal analyses---we now
provide a direct result on the sensitivity to scaling of the estimand of an OLS regression of an outcome
of the form $m(aY)$ on an arbitrary random variable $X$.

Specifically, suppose that $(X,Y) \sim Q$, for $Y \in [0,\infty)$ and $X \in \reals^J$, where the first element of $X$ is a constant. Consider the OLS estimand
$$\beta(a) = E_Q[XX']^{-1} E_Q[X m(aY)],$$
\noindent i.e. the population coefficient from a regression of $m(aY)$ on $X$. We assume that $E_Q[XX']$ is full-rank so that $\beta(a)$ is well-defined. Letting $\beta_j(a) = e_j'\beta(a)$ be the $j$\th{} element of $\beta(a)$, we will show
that $\beta_j(a)$ can be made to have arbitrary magnitude via the choice of $a$ if $\gamma_j \neq 0$, where 
$$\gamma = E_Q[XX']^{-1} E_Q[X \one[Y>0]]$$
\noindent is the coefficient from a regression of $\one[Y>0]$ on $X$.

\begin{prop} \label{prop: ols sensitivity}
Suppose that

\begin{enumerate}
    \item The function
    $m$ satisfies parts 1 and 2 of \Cref{prop: can get any value for ATE}.

    \item
    (Finite expectations) $E_Q[ \norm{X} ] <\infty $ and $E_Q[ \norm{X \log
    (Y)} \mid Y>0 ] < \infty$ .

\end{enumerate}

\noindent Then for every $j \in \{2,...,J\}$, $\beta_j(a)/\log(a) \to \gamma_j$ as $a \to \infty$. Moreover, if $\gamma_j \neq 0$ for some $j \in \{2,...,J\}$, then for every $\beta_j^* \in (0,\infty)$, there exists $a>0$ such that
$|\beta_j(a)| = \beta_j^*$. In particular $\beta_j(a)$ is continuous with $\beta_j(a) \to 0$
as $a \to 0$ and $|\beta_j(a)| \to \infty$ as $a \to \infty$.
\end{prop}

We note that \Cref{prop: ols sensitivity} implies that the OLS estimator for the  $j$\th{}
coefficient, $\hat\beta_j(a)$, will be arbitrarily sensitive to the choice of $a$ when the
corresponding extensive margin OLS estimator $\hat\gamma_j$, is non-zero. This follows
immediately from setting $Q$ to be the empirical distribution of $(Y_i,X_i)_{i=1}^N$ and
applying \Cref{prop: ols sensitivity} (note that part 2 of the Proposition is trivially
satisfied for the empirical distribution, since $X$ and $Y$ are both bounded over the
empirical distribution).

\paragraph{OLS Standard Errors.} We also show that as $a \to \infty$, the $t$-statistic for the OLS estimate $\hat\beta_j$ constructed using heteroskedasticity-robust standard errors converges to the $t$-statistic for $\hat\gamma_j$ (again using heteroskedasticity-robust standard errors). Formally, let 
$$\hat\Omega_\beta(a) = \left(\frac{1}{N} \sum_i X_i X_i'\right)^{-1} \left(\frac{1}{N} \sum_i X_i X_i' \hat{\epsilon}_i(a)^2\right)  \left(\frac{1}{N} \sum_i X_i X_i'\right)^{-1}$$
\noindent denote the estimator of the heteroskedasticity-robust variance matrix for $\hat\beta(a)$, where $\hat\epsilon_i(a) = m(aY_i) - X_i' \hat\beta(a),$ and $\hat\beta(a)$ is the OLS estimate of $\beta(a)$. The $t$-statistic for $\hat\beta_j(a)$ is then $\hat{t}_{\beta_j}(a) = \hat\beta_j(a) / \hat\sigma_{\beta_j}(a)$, where $\hat\sigma_{\beta_j}(a) = \sqrt{e_j' \hat\Omega_\beta(a) e_j} / \sqrt{N}$. Analogously, let 
$$\hat\Omega_\gamma = \left(\frac{1}{N} \sum_i X_i X_i'\right)^{-1} \left(\frac{1}{N} \sum_i X_i X_i' \hat{u}_i^2\right)  \left(\frac{1}{N} \sum_i X_i X_i'\right)^{-1}$$
\noindent be the heteroskedasticity-robust variance estimator for $\hat\gamma$, the OLS estimate of $\gamma$, where $u_i = \one[Y_i>0] - X_i'\hat\gamma$. The $t$-statistic for $\hat\gamma_j$ is then $\hat{t}_{\gamma_j} = \hat\gamma_j / \hat\sigma_{\gamma_j}$, where $\hat\sigma_{\gamma_j} = \sqrt{e_j' \hat\Omega_\gamma e_j} / \sqrt{N}$. 
\begin{prop}

\label{prop: OLS SEs}
Suppose that $\left(\frac{1}{N} \sum_i X_i X_i'\right)$ is full-rank and that $\hat\sigma_{\gamma_j} > 0$. If the function $m$ satisfies parts 1 and 2 of \Cref{prop: can get any value for ATE} and $\hat\gamma_j > 0$, then $\hat{t}_{\beta_j}(a) \to \hat{t}_{\gamma_j}$ as $a \to \infty$.    
\end{prop}

It follows that when the units of $Y$ are made large, the $t$-statistic for a treatment effect estimate for $m(Y)$ estimated using OLS will converge to the $t$-statistic for the OLS estimate of the extensive margin. \Cref{fig:tstats} shows that, indeed, the $t$-statistics for estimates using $\arcsinh(Y)$ in the \emph{AER} tend to be close to the $t$-statistics for the extensive margin, and tend to become even closer after rescaling the units by a factor of 100. 

\begin{proof}[Proof of \Cref{prop: ols sensitivity}]
Fix $j \in \{2,...,J\}$. Note that $\beta(0) = E_Q[XX']^{-1} E[X m(0)]$, is the coefficient from a regression of a
constant outcome $m(0)$ on $X$, and thus $\beta_1(0) = m(0)$ while $\beta_k(0) =0$ for $k
\geq 2$. Thus $\beta_j(0) = 0$. To complete the proof, we will first show that $\beta_j(a_n) = \gamma_j \log(a_n) + o(\log(a_n))$. Hence, if $\gamma_j >0$, then $|\beta_j(a)| \to
\infty$ as $a \to \infty$. We will then establish that $\beta_j(a)$ is continuous for $a \in [0,\infty)$. The fact that one can obtain any positive value for $|\beta_j(a)|$ then follows from the intermediate value theorem.

For ease of notation, let $\nu' = e_j'E_Q[XX']^{-1}$, so that $\beta_j(a) = E_Q[ \nu'X
m(aY) ]$. 

We first show that $\beta_j(a_n) = \gamma_j \log(a_n) + o(\log(a_n))$. Consider a sequence $a_n \to
\infty$, and assume without loss of generality that
$a_n > e$. Let $f_n(x,y) = \nu'x \cdot m(a_ny) / \log(a_n)$. Observe that $f_n(x,y) \to
\nu'x \cdot \one[y>0]$ pointwise, since $f_n(x,0) = \nu'x \cdot m(0) / \log(a_n) \to 0$,
while for $y>0$, $$f_n(x,y) = \nu'x \cdot \frac{m(a_ny)}{ \log(a_n) } = \nu' x \cdot
\frac{m(a_ny)}{\log(a_ny)} \frac{\log(a_n) + \log(y)}{ \log(a_n)} \to \nu 'x,$$
\noindent where we use the fact that $m(y)/\log(y) \to 1$ as $y \to \infty$. We showed in
the proof to \Cref{prop: log a rate} that $$ \frac{|m(a_n y)| }{ \log(a_n) }
\leq \kappa + 2 \cdot \one[y>0] \cdot (1 + |\log(y)|),$$
\noindent which implies that
$$|f_n(x,y)| \leq |\nu'x \cdot (\kappa + 2 \cdot \one[y>0] \cdot (1 + |\log(y)|))| =: \bar{f}(x,y).$$

\noindent Moreover, part 2 of the proposition implies that $\bar{f}(X,Y)$ is integrable. From the dominated convergence theorem, it follows that 

$$\frac{\beta_j(a_n)}{ \log(a_n) } = E_Q[f_n(X,Y)] \to E_Q[ \nu'X \one[Y>0] ] = \gamma_j.$$

\noindent Hence, we see that $\beta_j(a_n) = \gamma_j \log(a_n) + o(\log(a_n))$. It follows that $|\beta_j(a_n)| \to \infty$ when $\gamma_j \neq 0$. 

To complete the proof, we show continuity of $\beta_j(a)$. Fix $a \in [0,\infty)$, and
consider a sequence $a_n \to a$. Assume without loss of generality that $a_n < a+1$ for
all $n$. Let $f_n(x,y) = \nu'x \cdot m(a_n y)$. From the continuity of $m$, we have that
$f_n(x,y) \to \nu'x
\cdot m(ay)$ pointwise. We showed in the proof to \Cref{prop: can get any value for ATE} that there exists some $\kappa$ (not depending on $n$) such that 
$$|m(a_ny)| \leq \kappa + 2 \one[y>0] \cdot |\log(y)|.$$
\noindent Hence,
$$|f_n(x,y)| \leq |\nu'x \cdot (\kappa + 2 \one[y>0] |\log(y)| ) | .$$
\noindent Moreover, the bounding function is integrable over the distribution of $(X,Y)$ by part 2 of the proposition. Applying the dominated convergence theorem again, we obtain that

$$\beta_j(a_n) = E_Q[ f_n(X,Y) ] \to E_Q[ \nu' X \cdot m(aY) ] = \beta_j(a) , $$
\noindent as needed.
\end{proof}

\begin{proof}[Proof of \Cref{prop: OLS SEs}]
 Consider $a_n \to \infty$. Applying \Cref{prop: ols sensitivity} to the empirical
 distribution, we have that $\hat\beta(a_n) / \log(a_n)  = \hat\gamma +
 o(1)$. It follows
 that $$\frac{1}{\log(a_n)} \hat\epsilon_i(a_n) = \frac{m(a_nY_i)}{\log(a_n)} -
 \frac{\hat\beta(a_n)' X_i}{ \log(a_n) } = \one[Y_i > 0] - \hat\gamma' X_i + o(1)= \hat{u}_i +
 o(1).$$
 \noindent Since $\hat\Omega_n(a_n)$ is a continuous function of the $\hat{\epsilon}_i(a_n)^2$, we obtain that $\log(a_n)^{-2}
 \hat\Omega_\beta(a_n) \to \hat\Omega_\gamma$, and thus that $\log(a_n)^{-1}
 \hat{\sigma}_{\beta_j}(a_n) = \hat{\sigma}_{\gamma_j} + o(1)$. It follows that

 $$ \hat{t}_{\beta_j}(a_n) = \frac{\hat\beta_j(a_n) / \log(a_n)}{ \hat\sigma_{\beta_j}(a_n)/ \log(a_n) } = \frac{ \hat\gamma_j + o(1) }{ \hat\sigma_{\gamma_j} + o(1) } \to \frac{ \hat\gamma_j }{ \hat\sigma_{\gamma_j} } = \hat{t}_{\gamma_j} ,$$

 \noindent as needed.
\end{proof}

\section{Connection to structural equations models \label{sec: structural equations}}

Previous work has considered a variety of estimators for settings with zero-valued
outcomes beginning with structural equations models rather than the potential outcomes
model that we consider. These papers have reached different results, with some concluding
that regressions with $\arcsinh(Y)$ have the interpretation of an elasticity, and others
showing that they are inconsistent and advocating for other methods (e.g. Poisson
regression) instead. In this section, we interpret the results in those papers from the
perspective of the potential outcomes model, and show that these diverging conclusions
stem from different implicit assumptions about the potential outcomes, as well as a focus
on different causal parameters.

Before discussing specific papers, we first note that, broadly speaking, structural
equation models can be viewed as constraining the joint distribution of potential
outcomes. Observe that, for any pair of potential outcomes $(Y(1), Y(0))$, we can
represent them as $(Y(1, U), Y(0,U))$ for some function $Y(d, u)$ and individual-level
unobservable (or ``structural error'') $U$. The potential outcomes framework we work with
in this paper does not impose any functional form assumptions on $Y(d, u)$. Structural
equation models, on the other hand, tend to specify explicit functional forms for
$Y(d,u)$. In what follows, we consider the implicit restrictions placed on the potential
outcomes as well as the target estimand in work related work that starts with a structural
equations model.

\subsection{\citet{bellemare_elasticities_2020} and \citet{thakral2023estimates} \label{subsec:bellemare}}
\citet{bellemare_elasticities_2020} consider OLS regressions of the form\footnote{They
also consider specifications with additional covariates on the right-hand side, although
we abstract away from this for expositional simplicity.}
\begin{equation}
\arcsinh(Y) = \beta_0 + D \beta_1 + U. \label{eqn: arcsinh reg}
\end{equation}
\noindent Note that when $D$ is binary and randomly assigned, $D \indep (Y(1),Y(0))$, then
from the perspective of the potential outcomes model, the population coefficient $\beta_1$
is the ATE for $\arcsinh(Y)$. \citet{bellemare_elasticities_2020} instead consider the
interpretation of $\beta_1$ when \eqref{eqn: arcsinh reg} is treated as structural, i.e. if there are constant treatment effects of $D$ on $\arcsinh(Y)$. From
the perspective of the potential outcomes model, this amounts to imposing that the
potential outcomes $Y(d) := Y(d,U)$ take the form
\begin{equation}
\arcsinh( Y(d,U) ) = \beta_0 + d \beta_1 + U, \label{eqn: bellemare PO version}
\end{equation}

\noindent where the individual-level random variable $U$ takes the same value for all
values of $d$. Under \eqref{eqn: bellemare PO version}, we have that $$ \beta_1 = \arcsinh
( Y(1,U) ) - \arcsinh( Y(0,U) ) .$$

\noindent Since $\arcsinh(y) \approx \log(2y)$ for $y$ large, it follows that $\beta_1
\approx \log( Y(1,U) / Y(0,U) )$ when $Y(1,U)$ and $Y(0,U)$ are large. Thus,
\citet{bellemare_elasticities_2020} argue that $\beta_1$ approximates the semi-elasticity
of the outcome with respect to $d$ when the outcome is large. They likewise provide
similar results for the elasticity of $Y(d,U)$ with respect to treatment when treatment is
continuous. Their results thus imply that the ATE for $\arcsinh(Y)$ has a sensible
interpretation as a (semi-)elasticity when the structural equation for the potential outcomes given in
\eqref{eqn: bellemare PO version} holds.

It is worth emphasizing, however, that \eqref{eqn: bellemare PO version} will generally be
incompatible with the data when both $Y(1)$ and $Y(0)$ have point-mass at zero, and
$\beta_1 \neq 0$. Specifically, note that \eqref{eqn: bellemare PO version} implies that for all values of $U$,
$$\arcsinh(Y(1,U)) - \arcsinh(Y(0,U)) = \beta_1.$$

\noindent If $\beta_1 > 0$, for example, this implies that $\arcsinh(Y(1,U)) >
\arcsinh(Y(0,U))$, and hence $Y(1,U) > Y(0,U)$, since the $\arcsinh(y)$ function is
strictly increasing for $y \geq 0$. However, since $Y(0,U)\geq 0$ by assumption, this implies that $Y(1,U)>0$ with probability 1. Thus, the model in \eqref{eqn: bellemare PO
version} is incompatible with $P(Y(1)=0) >0$ if $\beta_1 > 0$. By similar logic, the model
is also incompatible with $P(Y(0) =0) >0$ if $\beta_1 < 0$. In settings where there is
point-mass at zero, the model that \citet{bellemare_elasticities_2020} show gives
$\beta_1$ an interpretation as a semi-elasticity will therefore typically be rejected by
the data. It is also worth noting that even if there are no zeros in the data, the model
in \eqref{eqn: bellemare PO version} will generally be sensitive to units, in
the sense that if \eqref{eqn: bellemare PO version} holds for $Y$ measured in dollars, it
will generally not hold when $Y$ is measured in cents. The validity of the interpretation
of $\beta_1$ as an elasticity thus depends on having chosen the ``correct'' scaling of the
outcome such that \eqref{eqn: bellemare PO version} holds.

Similar issues apply if we consider alternative transformations on the left-hand side of \eqref{eqn: arcsinh reg}. For example, \citet{thakral2023estimates} consider versions of \eqref{eqn: arcsinh reg} that replaces $\arcsinh(Y)$ with the power function $Y^k$. They then consider the implied ``semi-elasticities'' of the form $\eta(y_0) = \beta_1  /(k y_0^{k})$. The parameter $\eta(y_0)$ has the interpretation as a structural semi-elasticity when $d$ has a contant effect on $Y^k$. Specifically, if $D$ is continuous and the structural equation
\begin{equation}
    Y(d,U)^k = \beta_0 + d \beta_1 +U, \label{eqn: tt structural}
\end{equation}
\noindent holds, then $\eta(y_0) = \left( \frac{\partial}{\partial d} Y(d,U) \right) / Y(d,U)$ evaluated at $Y(d,U) = y_0$, so $\eta(y_0)$ corresponds to the semi-elasticity of $Y(d,U)$ with respect to $d$. However, as with \eqref{eqn: bellemare PO version}, \eqref{eqn: tt structural} is generally incompatible with settings in which $P(Y(d,U) = 0)$ for multiple values of $d$. For example, if $\beta_1 > 0$, then $Y(0,U) \geq 0$ implies that $Y(1,U) > 0$. Equation \eqref{eqn: tt structural}, which gives a causal interpretation to $\eta(\beta_0)$ as a semi-elasticity, will thus generally be incompatible with settings in which some units have $Y=0$ under multiple treatment statuses.

\subsection{\citet{cohn_count_2022}}
\label{subsec: cohn}
\citet{cohn_count_2022} consider structural equations of the form 
\begin{equation}
Y = \exp(\alpha + D \beta)  U. \label{eqn: cohn reduced form}
\end{equation}

\noindent When $E[ U \mid D] = 1$, they show that Poisson regression is consistent for
$\beta$, whereas regressions of $\log(1+Y)$ or $\log(Y)$ on $D$ may be inconsistent for
$\beta$.\footnote{We thank Kirill Borusyak for an insightful discussion on this topic. Relatedly, in an influential paper, \citet{silva_log_2006} consider the
structural equations model $Y_{i} = \exp(X_i' \beta)U_{i}$ where $E[U_i \mid X_i]=1$, and show
that Poisson regression consistently estimates $\beta$ while a regression using $\log$ on
the left-hand side does not, although they do not provide any formal results on log-like transformations.} Although \citet{cohn_count_2022} do not consider a potential
outcomes interpretation of $\beta$, we can give $\beta$ a causal interpreation if we impose that the potential outcomes take
the form
\begin{equation} 
Y(d, U) = \exp(\alpha + d \beta)  U(d), \label{eqn: cohn PO version}
\end{equation}

\noindent where $E[U(d) ]=1$. Under \eqref{eqn: cohn PO version}, it follows that $\exp(\beta) = E[Y(1)]/E[Y(0)]$, i.e. the parameter $\thetaPoisson$ considered in \Cref{subsec: interpretable
units}.\footnote{\citet{bellego_dealing_2022} also consider \eqref{eqn: cohn reduced
form}, but consider the more general class of identifying restrictions of the form $E[D
\log(U + \delta) ] = 0$, where $\delta$ is a tuning parameter.}

We note, however, that if one were instead to impose \eqref{eqn: cohn reduced form} with
the assumption that $E[ \log(U(d)) \mid D] =0$, then the regression of $\log(Y)$ on $D$ would
be consistent for $\beta$, whereas Poisson regression would generally be inconsistent for
$\beta$. Indeed, under the potential outcomes model in \eqref{eqn: cohn PO version} with
the assumption that $E[\log(U(d))]=0$, we have that $\beta = E[ \log(Y(1)) - \log(Y(0))]$,
the ATE in logs.\footnote{Note that the assumption that $E[\log(U)] =0$ implicitly
implies that $U > 0$, and thus $Y>0$.}

This discussion highlights that whether or not an estimator is consistent depends on the
specification of the \emph{target parameter}. Our results help to illuminate what
parameters can be consistently estimated by enumerating the properties
that identified causal parameters can (or cannot) have.

\subsection{Tobit models}
\label{sub:tobit}

An alternative structural approach is to explicitly model the extensive margin, a classic example of which is the Tobit model \citep{tobin_estimation_1958}. Following the discussion of Tobit models in \citet{AngristPischke(09)}, suppose there exist latent potential outcomes $Y^*(d) = \mu_d + U$, where $U \sim \Norm(0,\sigma^2)$ and $D \indep U$. The observed
potential outcome $Y(d)$ is then the latent potential outcome truncated at zero, $Y(d) = \max(Y^* (d), 0)$. We note that in this model, the
treatment has a constant additive effect of $\mu_1 -
\mu_0$ on the latent outcome, and the latent potential outcomes are assumed to be normally
distributed. 

Thanks to the parametric assumptions, the unknown parameters $\mu_1, \mu_0,
\sigma^2$ are identified and estimable via, e.g., maximum likelihood. As a result, the entire joint distribution of potential outcomes is identified, since this depends only on $(\mu_1,\mu_0,\sigma)$. This implies, in turn, that all of the possible target parameters considered in \cref{sec: recommendations} are point-identified. For example, under this model \[
\E[\log Y(d) \mid Y(1) > 0, Y(0) > 0] = \E\bk{\log \pr{\mu_d + U} \mid U > -\mu_1, U >
-\mu_0},
\]
where the right-hand side can be computed numerically since $U \sim \Norm(0, \sigma^2)$.
Thus, the intensive margin treatment effect in logs, $\theta_{\text{Intensive}}$, is actually point-identified under the Tobit
model.\footnote{Likewise, the intensive margin treatment effect in levels, $E[Y(1)-Y(0) \mid Y(1)>0,Y(0)>0]$ is simply $\mu_1 - \mu_0$.
}

It is worth nothing that unlike some of the models considered above, the Tobit model is consistent
with a nonzero extensive margin. However, the assumptions of normal errors and constant treatment effects on the latent index are restrictive. As discussed in \Cref{sec: recommendations}, imposing these assumptions is not necessary for identification if one is ultimately interested in, say, $\E[Y(1)-Y(0)]/\E
[Y(0)]$, and one can obtain bounds on the intensive margin effect without imposing these assumptions.\footnote{We note that the assumptions of the Tobit model imply (but are strictly stronger than) the assumption of rank preservation of the potential outcomes. However, rank preservation alone suffices to point identify $\E[\log Y(1) - \log Y(0)\mid Y(1) >
0, Y (0) > 0]$.} Moreover, as \citet{AngristPischke(09)} and
\citet{angrist_estimation_2001} point out, it is often not clear what the economic meaning of the latent potential
outcome $Y^*(d)$ is---if $Y(d)$ is earnings, for example, what is the meaning of having negative latent earnings ($Y^*(d)<0$)?

\section{Connection to two-part models\label{sec: 2pms}}
One approach recommended for settings with weakly-positive outcomes is to estimate a
two-part model \citep{mullahy_why_2023}. In this section, we briefly review two-part
models, and show that the marginal effects implied by these models do not correspond with
ATEs for the intensive margin without further restrictions on the potential outcomes.
Thus, while two-part models strike us as a reasonable approach if the goal is to model the
conditional expectation function of observed outcomes $Y$ given treatment $D$ (as in
\citet{mullahy_why_2023}), they will often not be appropriate if instead the goal is to
learn about a causal effect along the intensive margin.\footnote{We are particularly
grateful to John Mullahy for an enlightening discussion of this topic.}

The idea of a two-part model is to separately model the conditional distribution $Y \mid
D$ using (a) a first model for the probability that $Y$ is positive given $D$, $P(Y>0 \mid
D)$ (b) a second model for the conditional expectation of $Y$ given that it is positive,
$E[Y\mid D,Y>0]$. Common specifications include logit or probit for part (a), and a linear
regression of the positive values of $Y$ on $D$ for part b); see, e.g., \citet{belotti_twopm_2015}. After
obtaining estimates of the two-part model, it is common to evaluate the marginal effects
of $D$ on both parts, i.e. the implied values of
\begin{align*}
&\tau_{a} = P(Y>0 \mid D=1) - P(Y>0 \mid D=0) \\
&\tau_{b} = E[Y \mid Y>0,D=1] - E[Y\mid Y>0,D=0].
\end{align*}

We now consider how the parameters of the two-part model relate to causal effects in the
potential outcomes model. Suppose, for simplicity, that the two-part model is
well-specified, so that it correctly models $P(Y>0 \mid D)$ and $E[Y \mid Y>0,D]$. Suppose
further that $D$ is randomly assigned, $D \indep Y(1),Y(0)$. In this case, we have that
\begin{align*}
&\tau_{a} = P(Y(1)>0) - P(Y(0)>0) \\
&\tau_{b} = E[Y(1) \mid Y(1) >0  ] - E[Y(0) \mid Y(0) >0  ].
\end{align*}

\noindent From the previous display, we see that the marginal effect on the first margin,
$\tau_a$, has a causal interpretation: it is the treatment's effect on the probability
that the outcome is positive.

The interpretation of the marginal effect on the second margin, $\tau_b$, is more
complicated, however. For simplicity, suppose are willing to impose the ``monotonicity''
assumption discussed in \Cref{sec: recommendations}, $P(Y(1)=0,Y(0)>0) = 0$, so that
anyone with a zero outcome under treatment also has a zero outcome under control. Then,
letting $\alpha = P(Y(0) = 0 \mid Y(1) > 0 )$, we can write $\tau_{b}$ as
\begin{align*}
\tau_b &= (1-\alpha) E[ Y(1) \mid Y(1)>0,Y(0)>0] \\
&\quad\quad+ \alpha E[Y(1) \mid Y(1)>0, Y(0) = 0] - E[Y(0) \mid Y(1)> 0, Y(0) > 0 ] \\
&= \underbrace{E[Y(1) - Y(0) \mid Y(1)>0,Y(0)>0]}_{\text{Intensive margin effect}} \\ 
& \quad\quad+
\alpha \underbrace{\left( E[ Y(1) \mid Y(1)>0, Y(0) =0 ] - E[Y(1) \mid Y(1)>0, Y(0) >
0]\right)}_{\text{Selection term}},
\end{align*}

\noindent where the first equality uses iterated expectations, and the second re-arranges terms. 

The previous display shows that $\tau_b$ is the sum of two terms. The first is the ATE for
individuals who would have a positive outcome regardless of treatment status (similar to
$\theta_{\mathrm{Intensive}}$ in \Cref{sec: recommendations}, except using $Y$ instead of
$\log(Y)$). The second term is not a causal effect, but rather represents a selection
term: it is proportional to the difference in the average value of $Y(1)$ for ``compliers''
who would have positive outcomes only under treatment versus ``always-takers'' who would have
positive outcomes regardless of treatment status. In many economic contexts, we may expect
this selection effect to be negative. For example, we may suspect that individuals who
would only get a job if they receive a particular training have lower ability, and hence
lower values of $Y(1)$, than individuals who would have a job regardless of training
status. The marginal effect $\tau_b$ thus only has an interpretation as an ATE along the
intensive margin if either (a) there is no extensive margin effect ($\alpha=0$) or (b) we
are willing to assume that the selection term is zero. \citet{angrist_estimation_2001}
provided a similar decomposition (without imposing monotonicity), concluding that the
two-part model ``seems ill-suited for causal inference,'' at least without further
restrictions on the potential outcomes. See, also, \citet{mullahy_estimation_2001} for
additional discussion.

\section{Details on Lee bounds using IV in \citet{berkouwer2022credit}\label{sec:
berkouwer iv details}}

We now describe in detail our approach for constructing \citet{lee_training_2009}-type bounds in the IV setting of \citet{berkouwer2022credit}.

\paragraph{Estimating the instrument-complier distributions.} The first step is to estimate the
distribution of $Y(0)$ and $Y(1)$ for instrument-compliers. As shown in \citet{abadie_bootstrap_2002}, the CDF for $Y(1)$ for instrument-compliers at a point $y$ can be consistently estimated by using two-stage least squares to estimate the effect of treatment on the outcome $D_i 1[Y_i \leq y]$. The CDF for $Y(0)$ for instrument-compliers can analogously be obtained using the outcome $(D_i-1) 1[Y_i \leq y]$. We estimate these TSLS regressions using analogues to \eqref{eqn: 2sls berkouwer} (except replacing $\arcsinh(Y_i)$ with the outcomes just described) for all values of $y$ contained in the data. We thus obtain empirical estimates of the CDFs for instrument-compliers, $\hat{F}_{Y(d)}(y)$ for $d=0,1$.

\paragraph{Constructing bounds.} Note that if $U \sim U[0,1]$, then $Y(d) \sim F_{Y(d)}^{-1}(U)$, where $F_{Y(d)}^{-1}(u) := \inf\{ y \mid F_{Y(d)}(y) \geq u \}$. With this formulation in mind, \citet{lee_training_2009}'s bounds for $E[\log(Y(1)) \mid Y(1)>0,Y(0)>0]$ can be written as 
\begin{align*}
 E[ \log(F_{Y(1)}^{-1}(U)) \mid U \in [\theta_{NT}, \theta_{NT}+\theta_{AT}]] &\leq E[\log(Y(1)) \mid Y(1) >0, Y(0)>0 ] \\ &\leq E[ \log( F_{Y(1)}^{-1}(U) )\mid U \in [1-\theta_{AT}, 1]], \numberthis \label{eqn: lee bounds using U}   
\end{align*}\noindent where $\theta_{AT} = P(Y(1)>0,Y(0)>0)$, $\theta_{NT} = P(Y(1)=0,Y(0)=0)$, and $\theta_{C} = P(Y(1)>0,Y(0)=0)$. We estimate the bounds in \eqref{eqn: lee bounds using U} by plugging in the estimated CDFs for instrument-compliers described above, as well as the values of $\theta_{AT},\theta_{NT}, \theta_C$ implied by the estimated CDFs. We approximate the expectation over $U$ by taking the average over 100,000 uniform draws.\footnote{We note that in finite samples, the estimated CDF $\hat{F}_{Y(d)}(y)$ may be non-monotonic. Nevertheless, the inverse $\hat{F}^{-1}_{Y(d)}(u) := \inf\{y \mid \hat{F}_{Y(d)}(y) \geq u\}$ remains well defined.}  Finally, to compute the bounds on the treatment effect, we must estimate $E[Y(0) \mid Y(0) >0]$. To do this, we use the fact that 
$$E[Y(0) \mid Y(0) > 0] = E[F^{-1}_{Y(0)}(U) \mid U \in [\theta_{NT}+\theta_C,1] ] .$$
\noindent As before, we then estimate the right-hand-side in the previous display by plugging-in the estimated CDF for instrument-compliers, and simulating over 100,000 uniform draws. The Lee bounds for $\theta_{\text{Intensive}}$ are then obtained by subtracting the estimate of $E[Y(0) \mid Y(0) > 0]$ from the estimates of the lower and upper bounds in \eqref{eqn: lee bounds using U}. We estimate standard errors for the bounds using 1,000 draws from a non-parametric clustered bootstrap.\footnote{One complication that arises is that for some draws from the bootstrap distribution, the sign of the extensive margin can be the opposite of that in the original data. In our bootstrap procedure, we construct Lee-type bounds assuming monotonicity in whichever direction matches the bootstrapped data. The resulting bootstrap estimates of the bounds appear to be approximately normally distributed, but we think a formal theoretical evaluation of the bootstrap in this setting is an interesting topic for future work.}

\section{Appendix Tables and Figures}
 \renewcommand{\figurename}{Appendix Figure}
\setcounter{figure}{0}
\renewcommand{\tablename}{Appendix Table}
\setcounter{table}{0}

\begin{itemize}
    \item \Cref{tbl:selected-quotes-aer} contains information on the \emph{AER} papers
    discussed.
    \item \Cref{fig:tstats} shows how $t$-statistics change in the replication exercise. 
    \item \Cref{tbl:recale-by-100-aer-log1p} shows the analogue of 
    \cref{tbl:recale-by-100-aer} for $\log (1+Y)$. 
\end{itemize}

\begin{landscape}
\begin{table}
    \scriptsize
    \begin{tabularx}{\textwidth}{>{\raggedright\arraybackslash}p{3.7cm}>
{\raggedright\arraybackslash}p{2.7cm}>{\raggedright\arraybackslash}p{7cm}X}
\toprule
Paper & Interprets Units as \% & Original Units & Quote About Percents / Notes \\
\midrule
Azoulay et al (2019) & Yes & Publications (yearly) & ``In this case, coefficient estimates can be interpreted as elasticities, as an approximation.'' \\
Beerli et al (2021) & Yes & Patent applications (yearly) & ``The estimates thus reflect an approximate percentage increase.'' \\
Berkouwer and Dean (2022) & Yes & Weekly expenditure (dollars) & ``A 0.50 IHS reduction corresponds to a 39 percent reduction relative to the control group.'' \\
Cabral et al (2022) & Yes & Costs (dollar) per \$10K risk-adjusted covered payroll & Refers to estimates as ``the elasticities reported in panel A'' \\
Carranza et al (2022) & Yes & Hours worked (weekly) & ``Weekly earnings increase by 34\% (Table 1, column 3)'' \\
Faber \& Gauber (2019) & Yes & Municipality GDP (1000s of Pesos) & ``A one standard deviation increase in tourism attractiveness increases local manufacturing GDP by about 40 percent.'' \\
Hjort and Poulsen (2019) & Yes & KB per second & ``We find that cable arrival increases measured speed in connected locations, relative to unconnected locations, by around 35 percent'' \\
Johnson (2020) & Yes & Violations (monthly) & ``[T]he regression coefficient estimates the
ITT effect of a press release on the percent change in the number of violations. The point estimate (-0.18) is identical to the baseline estimate in percent terms -0.40/2.29 = 17.5\%).'' \\
Mirenda et al (2022) & Yes & Contract size (euros) & ``The amount of public funds awarded raises by 3.4 percent.'' \\
Norris et al (2021) & Yes & Criminal charges & ``We measure both the extensive margin (using a binary indicator for the outcome ever occurring) and the intensive margin (taking the inverse hyperbolic sine, IHS, of the number of times the outcome occurred, so the coefficient is interpreted as a percent change)'' \\
Ager et al (2021) & No interpretation & Wealth (1870 dollars) &  \\
Arora et al (2021) & No interpretation & Publications (yearly) &  \\
Bastos et al (2018) & No interpretation & Sales (yearly, euros) &  \\
Fetzer et al (2021) & No interpretation & Incidents (quarterly) &  \\
Moretti (2021) & No interpretation & Patents (yearly) &  \\
Rogall (2021) & No interpretation & Perpetrators &  \\
Cao and Chen (2022) & No & Rebellions per million population in 1600 & They compute $\exp
(\hat\beta) - 1$ and multiply by the baseline mean, then interpret this as the effect in levels \\
\bottomrule
\end{tabularx}

    \captionsetup[table]{labelformat=empty,skip=1pt}

    \caption{Papers in the \emph{AER} estimating effects for $\arcsinh(Y)$ with selected
    quotes}
    \label[appendixtable]{tbl:selected-quotes-aer}
   \floatfoot{\textit{Note:} this table lists papers in the \emph{AER} estimating
   treatment effects for $\arcsinh(Y)$. The second column classifies papers by whether
   they interpret the units of the treatment effect as a percent/elasticity, with
   categories ``yes'', ``no'', or ``no interpretation given.'' The third column describes the units of the outcome before applying the $\arcsinh$
   transformation, and the final column
   provides
   selected quotes and notes about the interpretation of the estimates. See \Cref{subsec: numerical illustration} for details.}
\end{table}
\end{landscape}

\begin{figure}[htb]
    \centering
    \includegraphics[width = 0.95 \linewidth]
    {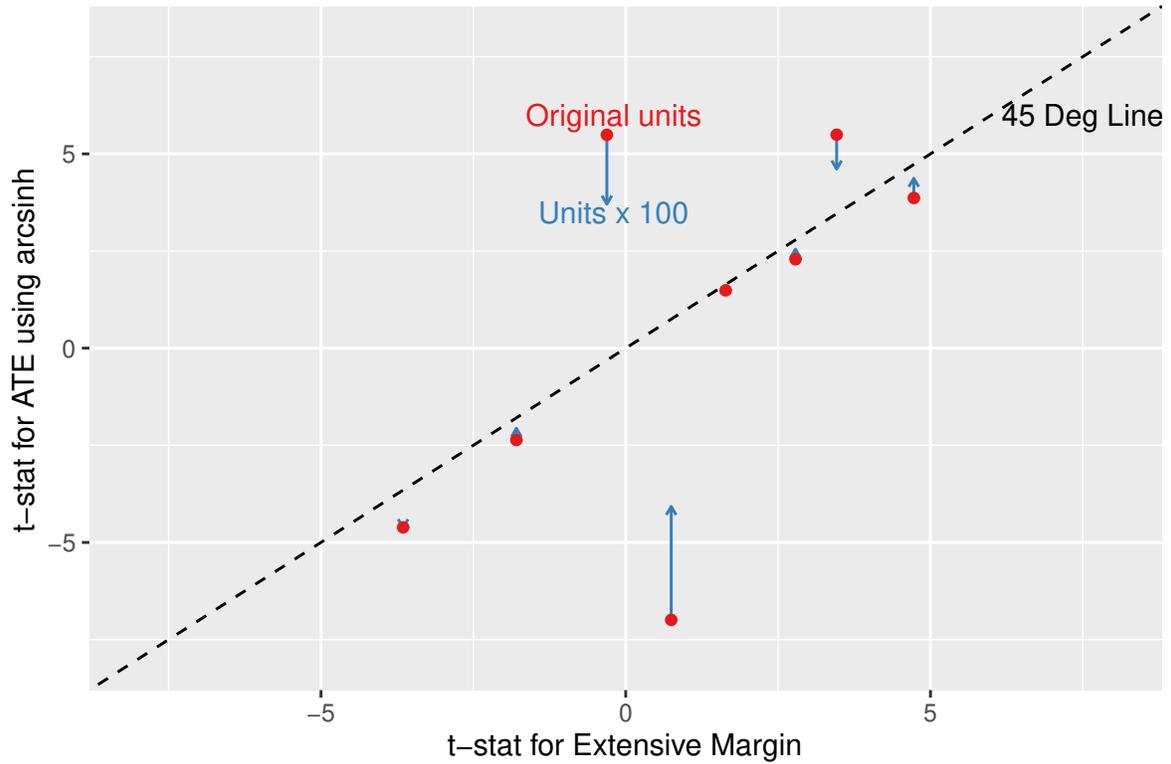}
    \caption{$t$-statistics for effect on $\arcsinh(Y)$, versus extensive margin $t$-statistic}
    \label[appendixfigure]{fig:tstats}
    \floatfoot{\textit{Note:} this table shows the $t$-statistic for the extensive margin effect on the $x$-axis, and the $t$-statistic for the treatment effect using $\arcsinh(Y)$ on the $y$-axis. The circle shows the $t$-statistic using the original units, whereas the arrow shows the change if we first multiply the units by 100 before applying the $\arcsinh$ transformation. We omit two papers where there is no extensive margin. The plot shows that the $t$-statistics are close to the 45 degree line when the extensive margin is not close to zero, and tend to become closer when the units are made larger.}
\end{figure}

\begin{table}[htb]
    \captionsetup[table]{labelformat=empty,skip=1pt}
    \begin{tabular}{lrrrR{0.08\textwidth}R{0.08\textwidth}}
\toprule
 & \multicolumn{2}{c}{Treatment Effect Using:} &  & \multicolumn{2}{c}{\thead{Change from \\ rescaling units:}} \\ 
 \cmidrule(lr){2-3} \cmidrule(lr){5-6}
Paper & $\log(1+Y)$ & $\log(1+100 \cdot Y)$ & Ext. Margin & Raw & \% \\ 
\midrule
Azoulay et al (2019) & 0.002 & 0.015 & 0.003 & 0.012 & 529 \\ 
Fetzer et al (2021) & -0.138 & -0.410 & -0.059 & -0.272 & 197 \\ 
Johnson (2020) & -0.139 & -0.408 & -0.057 & -0.269 & 194 \\ 
Carranza et al (2022) & 0.166 & 0.415 & 0.055 & 0.249 & 149 \\ 
Cao and Chen (2022) & 0.032 & 0.076 & 0.010 & 0.044 & 136 \\ 
Rogall (2021) & 1.109 & 2.015 & 0.195 & 0.906 & 82 \\ 
Moretti (2021) & 0.041 & 0.067 & 0.000 & 0.026 & 64 \\ 
Berkouwer and Dean (2022) & -0.412 & -0.484 & 0.010 & -0.072 & 17 \\ 
Arora et al (2021) & 0.110 & 0.111 & -0.001 & 0.001 & 1 \\ 
Hjort and Poulsen (2019) & 0.354 & 0.354 & 0.000 & 0.001 & 0 \\ 
\bottomrule
\end{tabular}

    \caption{Change in estimated treatment effects using $\log(1+Y)$ from re-scaling the outcome by a factor
    of 100 in papers published in the \emph{AER}}
    \label[appendixtable]{tbl:recale-by-100-aer-log1p}
    \floatfoot{Note: this table repeats the exercise in \cref{tbl:recale-by-100-aer} but replacing
    $\arcsinh(Y)$ with $\log (1+Y)$ as the outcome in the second column, and $\arcsinh(100Y)$ with
    $\log(1+100Y)$ in the third column. The fourth column shows the estimated extensive
    margin effect, which is identical to the fourth column of
    \cref{tbl:recale-by-100-aer}. The final two columns show the raw difference and
    percentage difference between the second and third columns. The rows are sorted based
    on the percentage differences. Among the papers surveyed, which by construction report at least one specification using $\arcsinh(Y)$, \citet{arora2021knowledge,fetzer2021security,moretti2021effect,rogall2021mobilizing}
    also report specifications that contain $\log(1+Y)$ on the left-hand side, and
    \citet{johnson2020regulation} reports a specification with $\log(c + Y)$ on the left-hand side, where $c$ is the first nonzero percentile of the distribution of the
    observed outcome variable.}
\end{table}

 \clearpage

{\singlespacing
\bibliographystyle{aer}
\bibliography{Bibliography}}

\end{document}